\documentclass[10pt]{article}
\usepackage[a4paper,margin=2.8cm]{geometry}
\usepackage[T1]{fontenc}
\usepackage{times}
\usepackage{enumerate}
\usepackage{amsmath,amssymb,amsthm,amsfonts}
\usepackage{graphicx}
\graphicspath{ {./images/} }
\usepackage{authblk}
\usepackage{color}
\usepackage[all,cmtip]{xy}
\usepackage{cite}
\usepackage{bm}
\usepackage{hyperref}
\usepackage{marginnote}
\usepackage[table]{xcolor}
\usepackage{dsfont}
\usepackage{stmaryrd}
\usepackage{multirow}
\usepackage{algorithm}
\usepackage{algorithmic}
\usepackage{subcaption}
\usepackage{placeins}
\usepackage{longtable}

\newcommand{\Esp}{\mathbb{E}}

\newcommand{\Prob}{\mathbb{P}}

\newcommand{\IPB}{\textnormal{IPB}}
\newcommand{\ITE}{\textnormal{ITE}}
\newcommand{\AIPB}{\textnormal{AIPB}}
\newcommand{\opt}{\textnormal{opt}}
\newcommand{\R}{\mathbb{R}}
\newcommand{\N}{\mathbb{N}}
\newcommand{\norm}[1]{\vert\vert#1\vert \vert}

\DeclareMathOperator{\supp}{supp}

\DeclareMathOperator{\CATE}{CATE}
\DeclareMathOperator{\Multinomial}{Multinomial}
\DeclareMathOperator{\Binom}{Binom}
\DeclareMathOperator{\Ber}{Ber}

\theoremstyle{plain} 
	\newtheorem{theorem}{Theorem}[section]
	\newtheorem{corollary}[theorem]{Corollary}
	\newtheorem{proposition}[theorem]{Proposition}
	\newtheorem{lemma}[theorem]{Lemma}
	\theoremstyle{definition} 
	\newtheorem{definition}[theorem]{Definition}
	
	\newtheorem{remark}[theorem]{Remark}

\numberwithin{equation}{section}

\title{Optimal treatment regimes for the net benefit of a treatment}

\author[1]{François Petit}
\author[3]{Gérard Biau}
\author[1,2]{Rapha\"el Porcher}
\affil[1]{Universit\'e Paris Cit\'e and Universit\'e Sorbonne Paris Nord, Inserm, INRAE, Center for Research in Epidemiology and StatisticS (CRESS), F-75004 Paris, France}
\affil[2]{Centre d'\'Epid\'emiologie Clinique, Assistance Publique-H\^opitaux de Paris, H\^otel-Dieu, F-75004, Paris, France}
\affil[3]{Sorbonne Université, CNRS, LPSM, Institut universitaire de France, F-75005 Paris, France}

\begin{document}

\maketitle

\begin{abstract}
   We develop a mathematical framework to define an optimal individualized treatment rule (ITR) within the context of prioritized outcomes in a randomized controlled trial. Our optimality criterion is based on the framework of generalized pairwise comparisons. We propose two approaches for estimating optimal ITRs on a pairwise basis. The first approach is a variant of the k-nearest neighbors algorithm. The second approach is a meta-learning method based on a randomized bagging scheme, which enables the use of any classification algorithm to construct an ITR. We investigate the theoretical properties of these estimation procedures, evaluate their performance through Monte Carlo simulations, and demonstrate their application to clinical trial data.
\end{abstract}

\section{Introduction}

Personalized medicine aims to tailor treatments to individual patient characteristics. Typically, this involves constructing individualized treatment rules (ITRs) that optimize population-level outcomes if all individuals were to adhere to these rules. Most methods to derive ITRs are based on a single outcome, though clinical decisions,  as well as regulatory evaluations, usually weigh several outcomes, generally at least benefits and harms, in the so-called benefit-risk assessment \cite{Willan97}, but also different measures of benefit, including clinical outcomes and patient-reported outcomes, for instance.

For example, the Third International Stroke Trial (IST-3) that evaluated thrombolysis for acute ischemic stroke \cite{ist3} primarily targeted the proportion of patients who were alive and independent at six months, as assessed by the Oxford Handicap Score (OHS) \cite{Bamford1990}. A secondary outcome was health-related quality of life at six months, measured using the EuroQol scale \cite{Herdman2011}. These outcomes reflect complementary dimensions of the treatment effect, prompting the pursuit of optimal ITRs that accommodate multiple endpoints simultaneously.

In practice, distinct outcomes are often analyzed independently, neglecting their interdependencies \cite{Lynd2004}. When considered jointly, a common strategy to develop ITRs is to optimize a primary outcome under constraints for a secondary one. This is particularly suitable when the secondary outcome reflects treatment harms \cite{Wang2018,Huang2020}. Nevertheless, this constrained approach becomes inappropriate when secondary outcomes cannot reasonably serve as constraints. Therefore, our objective was to develop a methodology for treatment personalization that explicitly incorporates multiple outcomes, including those that cannot be expressed as simple constraints.

Our proposed methodology builds upon concepts originally developed for analyzing randomized clinical trials with multiple endpoints. Specifically, we draw on the framework of prioritized outcomes introduced by \cite{Finkelstein99,Buyse:2010aa,Pocock2012}, which extends the average treatment effect paradigm to quantify treatment impacts across multiple prioritized endpoints. The causal estimands these methods target have been formally characterized by \cite{Luo:2017aa}, who provided theoretical foundations for their estimation through $U$-statistics. 

Let us provide a more detailed description of the Generalized Pairwise Comparison (GPC) approach introduced for estimating the \textit{net benefit} in \cite{Buyse:2010aa}. The net benefit is defined as the difference between the probability that a randomly selected participant in the experimental group achieves a better outcome than a randomly selected participant in the control group and the probability of the reverse. This approach involves hierarchically comparing multiple outcomes between paired patients, with each pair consisting of one treated and one control participants. Comparisons within each pair begin with the highest-priority outcome and only proceed to subsequent outcomes if no distinction is evident at the current level. Once a clear difference is observed at any outcome level, subsequent outcomes are disregarded for that pair. Each pair is then classified as \textit{favorable} if the treated patient outperforms the control, \textit{unfavorable} if the opposite occurs, or \textit{neutral} in the case of a tie. Ultimately, the net benefit quantifies the overall advantage of the treatment and is typically estimated using $U$-statistics based on these pairwise comparisons.

While the GPC approach described above effectively estimates the net benefit at a population level, extending it to personalized treatment settings requires additional considerations. 
In particular, constructing personalized treatment rules incorporating prioritized outcomes entails (1) defining an appropriate criterion of optimality for treatment strategies and (2) proposing methods for estimating such optimal rules. 

Personalizing treatment requires comparing potential outcomes under different treatments within the same individual rather than across distinct individuals. More precisely, when constructing a causal framework for treatment personalization, one would like to compare the potential outcomes $Y(0)$ and $Y(1)$ of a patient with characteristics $X$ that is, comparing the outcome $Y(0)$, had the patient received the control treatment, with the outcome $Y(1)$, had the same patient received the experimental treatment. Thus, defining a relevant causal quantity that is identifiable is critical. A frequent choice for this causal estimand is the conditional average treatment effect (CATE) defined as $\CATE(X)=\mathbb{E}[Y(1)-Y(0)|X]$, which remains causally identifiable thanks to the linearity of the expectation. An estimation of the CATE immediately yields an estimation of an optimal individualized treatment rule.

Conversely, directly applying non-linear functions $f$ (as in the GPC framework’s hierarchical comparisons) to potential outcomes---such as computing $\mathbb{E}[f(Y(1), Y(0)) | X]$---is problematic. This approach typically compromises causal identifiability due to the non-linearity of $f$.

To address this challenge, leveraging the GPC framework in stratified randomized trials, we observe that patients within a given stratum are assumed to represent identical, independent copies, and note that stratified GPC methods, where patients are grouped into categorical strata, already reflect a basic form of treatment personalization based on net benefit signs within strata.

Formally, considering a pair formed by the individuals $(X,Y(0),Y(1))$ and $(U,V(0),V(1))$,  with characteristics $X$ (resp. $U$) and potential outcomes $Y(0)$, $Y(1)$ (resp. $V(0)$,$V(1)$), we define:
\[
\Delta_{(r_0,r_1)}(X,U)=\mathbb{E}[\sigma(Y(0),V(1))|X,U].
\]
However, for personalized medicine, this quantity lacks direct relevance. Therefore, we identify an individual patient characterized by $x$ as the pair $(x,x)$ and define the individualized pairwise benefit (IPB) as:
\[
\IPB_{(r_0,r_1)}(x)=\Delta_{(r_0,r_1)}(x,x),
\]
representing the probability that an individual with characteristics $x$ receiving the experimental treatment achieves a "better" outcome compared to an identical, independent control individual. The IPB is well-defined under mild conditions and coincides with CATE for binary outcomes. 

The individualized pairwise benefit $\IPB_{(r_0,r_1)}$ allows to define a notion of optimality for individualized treatment rules aiming at optimizing treatment decisions with respect to a hierarchy of outcomes. In this setting, the estimation of optimal treatment strategies reduces to accurately estimating $\IPB_{(r_0,r_1)}$.

While the IPB provides a useful measure of optimality, alternative criteria also exist for evaluating treatment strategies. We explore several such causal estimands, comparing their interpretations and relating them to the net benefit concept from \cite{Buyse:2010aa}. Though different criteria lead to the same optimal treatment rule, they may differ when ranking suboptimal rules. The influence of the choice of the causal estimand used to assess treatment effect with prioritized outcome has been extensively studied in \cite{Even2025}, where they introduce a novel individual-level causal effect measure to improve treatment effect evaluation based on the win-ratio.

This paper is structured as follows. Section \ref{sec:GPC} reviews prioritized outcomes and the generalized pairwise comparison approach. Section \ref{sec:IPC} introduces our proposed optimality notion for ITRs using the IPB. We then propose two estimators for the $\IPB$ in the context of randomized clinical trials. The first is a nearest-neighbors-based, two-sample conditional $U$-statistics estimator. The second is a classification-based approach, where a classifier is trained to predict the type of pair (favorable, neutral, or unfavorable), which is then used to estimate the $\IPB$. In Section \ref{sec:theoritical_result}, we establish sufficient conditions for the convergence of the conditional $U$-statistics estimator. We also introduce and prove the convergence of a bagging method to mitigate the large memory footprint of the classification-based method. Section \ref{sec:simul} presents a finite-sample analysis via Monte-Carlo simulation of the classification-based estimator of the $\IPB$. Finally, in Section \ref{sec:data}, we illustrate our methodology by applying it to data from IST-3 to determine optimal treatment recommendations for prescribing intravenous recombinant tissue plasminogen activator to patients with ischemic stroke within three hours of symptom onset. Finally, Section \ref{sec:discussion} discusses broader implications and future research directions.

\section{Generalized pairwise comparisons} \label{sec:GPC}

Let us consider a randomized clinical trial where $n$ individuals have been randomized to the experimental treatment group $E$ and $m$ individuals have been randomized to the control group $C$. We consider that participants in the group $C$ are planned to receive a treatment option $A = 0$ and those in the group $E$ the other treatment option, $A = 1$. For an individual in the trial, we denote by $X$ (resp. $U$) a vector of baseline (pre-randomization) covariates if s/he is assigned to the control (resp. experimental) group and by $Y$ (resp. $V$) a vector of outcomes.

Generalized pairwise comparisons rely on forming pairs of individuals, one from group $C$ and one from group $E$, ordering their outcomes and assigning a score depending on that ordering. Consider a pair consisting of the individual $i$ of group $C$ and individual $j$ of group $E$. Their pair of outcomes is $(Y_i,V_j)$, and the corresponding score is defined as follow:
\begin{equation*}
\sigma_{ij}=
\begin{cases}
 +1 \; \textnormal{if} \;  V_j \succ Y_i \\
 0 \; \textnormal{if} \;  V_j \bowtie Y_i\\
 -1 \; \textnormal{if} \;  V_j \prec Y_i\\
\end{cases}
\end{equation*}
where the symbols '$\succ$' and '$\prec$' denote superiority and inferiority of outcomes, respectively, and $\bowtie$ indicates that the comparison is inconclusive. To allow more complex calculations, we will later use the more general notation $\sigma(Y_i, V_j)$ to generalize the definition of $\sigma_{ij}$. In the simple case of a single binary or continuous outcome, $Y$ and $V$ are scalars, and the ordering simply relates to the natural ordering of $\mathbb{R}$, adding the consideration of whether larger values of $Y$ (resp. $V$) are beneficial or detrimental to the individuals. The average treatment effect is then expressed in terms of net benefit---or proportion in favor of treatment---$\Theta$ estimated as:
\begin{equation*}
\hat\Theta = \frac{\sum_{i=1}^m\sum_{j=1}^n \sigma_{ij}}{n \times m}.
\end{equation*}

This definition of scores extends to more complex settings, such as time-to-event outcomes, or if requiring $|V_j - Y_i|$ being larger than a threshold $\tau$ for $\sigma_{ij}$ being different from 0 (i.e., $|V_j - Y_i| \le \tau$ leads to $V_j \bowtie Y_i$) \cite{Buyse:2010aa,Peron:2018aa}. The key idea underlying generalized pairwise comparisons is to equip the space of outcomes with an order.
When several outcomes are jointly considered, that is when $Y=(Y^1,\ldots,Y^d) \in \mathbb{R}^d$ (resp. $V=(V^1,\ldots,V^d) \in \mathbb{R}^d$) is a vector, a possible way of specifying an order is to first provide an order $\prec^k$ for the $k$-th coordinate space and second form the lexicographic order on $\mathbb{R}^d$ induced by the $\prec^k$'s (see tables \ref{tab:score} and \ref{tab:score_multi} for examples of such orders) .

Generalized pairwise comparisons allow for quantitative assessment of the benefit of an intervention, compared to a control \cite{Peron:2015aa,Buyse:2021aa} using a hierarchy of outcomes where the lexicographic order entails deciding which of those outcomes to prioritize. In the case of a prioritized outcome with two outcomes, the table~\ref{tab:score} presents how a pair of individuals would be scored with binary outcomes, assuming a first (higher priority) and second (lower priority) efficacy outcomes, for which a value of 1 would be beneficial to the patient. More complex orders can be considered, making it virtually possible to cover any situation. A second example with a numerical outcome is given in table~\ref{tab:score_multi}

\begin{table}[htb]
	\centering
	\caption{Scoring $\sigma_{ij}$ of a pair $(i,j)$ of individuals taken from the control ($C$) and experimental ($E$) groups, respectively, with binary outcomes.}\label{tab:score}
	\begin{tabular}{lcclccc}
		\hline
		\multicolumn{3}{l}{Outcome with higher priority} & \multicolumn{3}{l}{Outcome with lower priority}  & \\ \cline{1-6}
		$(Y^{1}_i,V^{1}_j)$ & Ordering & Label & $(Y^{2}_i,V^{2}_j)$ & Ordering & Label & $\sigma_{ij}$ \\ \hline
		$(0,1)$ &  $V^{1}_j  \succ Y^{1}_i$ & Favorable & Any & --- & Not considered & $+ 1$ \\
		$(1,0)$ & $V^{1}_j \prec Y^{1}_i$ & Unfavorable & Any & --- & Not considered & $-1$\\
		$(1,1)$ or $(0,0)$ & $V^{1}_j \bowtie Y^{1}_i$ & Tie/neutral & $(0,1)$ & $V^{2}_j \succ Y^{2}_i$ & Favorable & $+ 1$\\
		$(1,1)$ or $(0,0)$ & $V^{1}_j \bowtie Y^{1}_i$ & Tie/neutral & $(1,0)$ & $V^{2}_j \prec Y^{2}_i$ & Unfavorable & $-1$\\
		$(1,1)$ or $(0,0)$ & $V^{1}_j \bowtie Y^{1}_i$ & Tie/neutral & $(1,1)$ or $(0,0)$ & $V^{2}_j \bowtie Y^{2}_i$ & Tie/neutral & $0$\\
		\hline
	\end{tabular}
\end{table}

\begin{table}[htb]
	\centering
	\caption{Scoring $\sigma_{ij}$ of a pair $(i,j)$ of individuals taken from the control ($C$) and experimental ($E$) groups. The outcome is a vector constituted of a binary outcome and a continuous outcome. The $i^{th}$ component of this vector is denoted $Y^i$ for patients in the control group and $V^i$ for patients in the experimental group. The random variables $Y^{1}_i,V^{1}_j$ encodes the  binary outcome and $Y^{2}_i,V^{2}_j$ encodes the continuous outcome with threshold $\delta$.}\label{tab:score_multi}
	\begin{tabular}{lcclccc}
		\hline
		\multicolumn{3}{l}{Outcome with higher priority} & \multicolumn{3}{l}{Outcome with lower priority}  & \\ \cline{1-6}
		$(Y^{1}_i,V^{1}_j)$ & Ordering & Label & $(Y^{2}_i,V^{2}_j)$ & Ordering & Label & $\sigma_{ij}$ \\ \hline
		$(0,1)$ &  $V^{1}_j \succ Y^{1}_i$ & Favorable & Any & --- & Not considered & $+ 1$ \\
		$(1,0)$ & $V^{1}_j \prec Y^{1}_i$ & Unfavorable & Any & --- & Not considered & $-1$\\
		$(1,1)$ or $(0,0)$ & $V^{1}_j \bowtie Y^{1}_i$ & Tie/neutral & $Y^{2}_i-V^{2}_j<-\delta$ & $V^{2}_j \succ Y^{2}_i$ & Favorable & $+ 1$\\
		$(1,1)$ or $(0,0)$ & $V^{1}_j \bowtie Y^{1}_i$ & Tie/neutral &  $Y^{2}_i-V^{2}_j>\delta$ & $V^{2}_j \prec Y^{2}_i$ & Unfavorable & $-1$\\
		$(1,1)$ or $(0,0)$ & $V^{1}_j \bowtie Y^{1}_i$ & Tie/neutral & $ \vert Y^{2}_i-V^{2}_j \vert \leq \delta$ & $V^{2}_j \bowtie Y^{2}_i$ & Tie/neutral & $0$\\
		\hline
	\end{tabular}
\end{table}

\section{Individualized pairwise comparisons} \label{sec:IPC}

\subsection{General idea}

This subsection explains how generalized pairwise comparisons can be utilized to personalize treatments in the context of prioritized outcomes, using data from randomized control trials (RCTs).

Our approach builds on the use of generalized pairwise comparisons with stratification, as described in \cite{Buyse:2010aa}. Traditionally, this method is applied within strata defined by one or a few categorical variables. In such cases, deriving an optimal treatment rule involves simply recommending the treatment according to the sign of the net benefit within each stratum.

Here, we extend this approach to handle stratification based on an arbitrary set of predictors, including continuous ones, and formalize the construction of optimal individualized treatment rules (ITRs) in such scenarios. Specifically, we define a metric that establishes a notion of optimality for ITRs aimed at optimizing prioritized outcomes.

In practice, we assume that a score, denoted by \(\sigma\), based on prioritized outcomes is provided. Our objective is to construct an ITR that is \textit{pairwise optimal} (see subsection \ref{subsec:opt} for a precise definition).
Here, our goal is not to evaluate the treatment effect using RCT data but rather to leverage these data to develop optimal treatment rules.

Let \(X_i\) and \(U_j\) represent the covariates of patients in the control and experimental arms, respectively, and \(Y_i\) and \(V_j\) their corresponding prioritized outcomes. We consider patient pairs \((X_i, U_j, \sigma(Y_i, V_j))_{1 \leq i \leq n, 1 \leq j \leq m}\), where each pair falls into one of the following categories:
\[
\begin{aligned}
\begin{cases}
(x_i, u_j, -1) & \text{\parbox[t]{10cm}{if the patient in the control arm had a more favorable outcome (unfavorable pair)}} \\
(x_i, u_j, 0)  & \text{if the patients had similar outcomes (neutral pair)} \\
(x_i, u_j, 1)  & \text{\parbox[t]{10cm}{if the patient in the experimental arm had a more favorable outcome (favorable pair).}} \\
\end{cases}
\end{aligned}
\]
To construct the ITR, we estimate a function \(\widehat{\Delta}\) that approximates \(\Delta\), which predicts the difference between the probability of a favorable pair and the probability of an unfavorable pair, given the covariates. We then restrict \(\widehat{\Delta}\) to the diagonal—i.e., \(\widehat{\Delta}(x, x)\)—and recommend the experimental treatment to a patient if \(\widehat{\Delta}(x, x) > 0\).

We propose two methods to estimate \(\Delta\): (1) direct estimation using regression techniques, or (2) framing the task as a classification problem to predict the label (unfavorable, neutral, favorable) of a pair \((x, u)\) and subsequently estimating \(\Delta\) as the difference between the probabilities of favorable and unfavorable pairs.

\subsection{Policy optimality}\label{subsec:opt}
Let $\mathcal{X}$ be the space of covariates and $\mathcal{Y}$ be the space of outcomes. We assume that $\mathcal{X}$ is a metric space which distance is denoted $d_\mathcal{X}$. We assume that $\mathcal{Y}$ is endowed with a binary relation $\prec$. In this section, we ignore measure theoretic issues and refer the interested reader to the Appendix \ref{app:existence} were the more mathematical aspects are taken care of.

We assume that we are given two independent and identically distributed random variables $X$ and $U$. Following Neyman-Rubin causal model, we consider that a patient with characteristics $X$ (resp. $U$) with observed outcome $Y$ (resp. $V$) has two potential outcomes $Y(0)$ (resp. $V(0)$) and $Y(1)$ (resp. $V(1)$) representing the outcome s/he would achieve if, possibly contrary to fact, s/he had received treatment option $A = 0$ or $A = 1$ respectively. 

An ITR is a map $r \colon \mathcal{X} \to \{0;1\}$ which assigns to each patient a treatment option. We assume that we are given two ITRs $r$ and $s$ and that the patients of population $X$ are treated accordingly to the ITR $r$ while those of population $U$ are treated according to the ITR $s$.

We define the potential outcomes under rule $r$ and $s$ via the following formula
\begin{align*}
    Y(r)&=r(X)Y(1)+(1-r(X))Y(0),\\
    V(s)&=s(U)V(1)+(1-s(U))V(0)
\end{align*}
and further assume that $Y=Y(r)$ and $V=V(s)$ (consistency assumption).

We define the following causal quantity, which we term the \textit{pairwise benefit}.
\begin{equation*}
    \Delta_{(r, s)} (X,U)=\Prob(  V(s) \succ Y(r)  \vert X,U)- \Prob(V(s) \prec Y(r)  \vert X,U).
\end{equation*}
Since our aim it to make recommendations at the individual level it is natural to introduce the \textit{individualized pairwise benefit} (IPB) of the pair $(r,s)$ which, is well-defined under mild technical assumptions (see Appendix \ref{app:existence})
\begin{equation*}
    \IPB_{(r,s)}(x)=\Delta_{(r, s)}(x,x).
\end{equation*}
The \textit{average individualized pairwise benefit} (AIPB) of the pair $(r,s)$ is the quantity
\begin{equation}\label{eq:AIPB}
    \AIPB_{(r,s)}=\Esp(\IPB_{(r,s)}(X)).
\end{equation}
We say that an ITR $s$ is \textit{pairwise optimal} if and only if for every ITR $r$, $\AIPB_{(r,s)}\geq 0$.

\begin{remark}
    There is another possible notion of optimality based on pairwise comparison that is closer in spirit to the proportion in favor of treatment. We compare it to pairwise optimality in Appendix \ref{sec:propinfavor}. Under suitable assumptions, they lead to the same optimal rules. Nonetheless they do no rank suboptimal rules in the same order. An issue with this alternative  notion of optimality is that it depends on the probability of drawing a given pair among pairs of the form $(x,x)$ (see Equation \eqref{eq:gamma_disc} and Theorem \ref{thm:sqr_density}).
\end{remark}

A randomized trial (or an observational study) can be interpreted in the above setting as follows. It corresponds to the comparison of the two constant ITRs $r_0\colon \mathcal{X} \to \{0;1\}$, assigning the treatment labeled zero to every patients, and $r_1\colon \mathcal{X} \to \{0;1\}$, assigning the treatment labeled one to every patient. 

\begin{proposition}\label{prop:opt}
The ITR defined by
\begin{equation*}
s^{\opt}:=\mathds{1}_{\{\IPB_{(r_0,r_1)}>0\}}
\end{equation*}
is pairwise optimal.
\end{proposition}

We now relate the individualized pairwise benefit and the individualized treatment effect (or conditional average treatment effect) which forms the basis of many approaches to develop ITRs in the classical counterfactual model setting \cite{Kunzel2019}. Let us recall that $\ITE(x)=\Esp(Y(1)-Y(0)|X=x)$.

\begin{proposition}\label{prop:comp_bin}
    Assume that the potential outcomes considered are binary (i.e., $\mathcal{Y}=\{0,1\}$) and that the relation $\prec$ is the usual order on $\{0;1\}$. Then $\IPB_{(r_0,r_1)}(x)=\ITE(x)$.
\end{proposition}

\subsection{Estimation}
The  aim of this section is to provides methods to construct pairwise optimal ITR.   Again the idea underlying our approach is to learn a contrast between the outcome of patients in each arm of the trial and deduce an estimation of the optimal ITR. This problem can be approached via regression or classification methods. Indeed, one can either 
\begin{enumerate}[(i)]
\item estimate $\Delta_{(r_0,r_1)}$ using regression techniques then deduce $\IPB_{(r_0,r_1)}$ and obtain an estimate of the optimal rule. 
\item Use multi-class classification techniques to learn a classifier that predicts for each pair $(X,U)$ the pseudo-outcomes corresponding to the status of the pair (favorable, neutral, unfavorable) and deduce from it an estimate of $\Delta_{(r_0,r_1)}$ which in turn provides an estimates of the optimal ITR.
\end{enumerate}

We recall the data generating mechanism that we are considering.

\subsubsection{Model of the data generating mechanism}

In this section, we describe how we model a RCT. Our approach is similar to the one of \cite{Buyse:2010aa}. We assume that we are given two groups of patients respectively denoted $C$ and $E$ sampled from the same population. The patients of the group $C$ are exposed to the control treatment while those of the group $E$ are exposed to an experimental treatment. 
We write $(X_i,Y_i(0),Y_i(1), Y_i)$ for the characteristics, potential outcomes and outcome of a patient of the group $C$ and, write $(U_i,V_i(0),V_i(1),V_i)$ for the characteristics, potential outcomes and outcome of a patient of the group $E$.
For the sake of simplicity we assume that the patients of group $C$ are numbered from $1$ to $m=|C|$ and the patients of group $E$ are numbered from $1$ to $n=|E|$ with $N=m+n$. Throughout the rest of the paper, we assume the following

\begin{enumerate}[(H1)]
    \item The random variables  $(X_i,Y_i(0),Y_i(1))$ and $(U_j,V_j(0),V_j(1))$ are iid, $1 \leq i \leq m$, $1 \leq j \leq n$. Note that the outcome is not part of this tuples.
    \item The random variables $(X_i,Y_i(0),Y_i(1),Y_i)$ for $1 \leq i \leq m$ are iid,
    \item The random variables $(U_j,V_j(0),V_j(1),V_j)$ for $1 \leq j \leq n$ are iid,
    \item $Y_i=Y_i(0)$ and $U_j=U_j(1)$.
\end{enumerate}

\begin{proposition}
    Under Assumptions $\textnormal{(H1)}$ to $\textnormal{(H4)}$, the random variables $(X_i,Y_i)$ and $(U_j,V_j)$ are independent.
\end{proposition}

\subsubsection{Conditional U-statistics}

In this section, we provide a method to estimate $\Delta_{(r_0,r_1)}$ via a conditional version of the Wilcoxon-Mann-Whitney statistic. This first method is in the spirit of the previous works on pairwise comparison which relies on U-statistics such as the Wilcoxon-Mann-Whitney statistics. Indeed our estimator is an extension of the notion of one sample conditional U-statistics due to W. Stute \cite{Stute94} to the case of 2-samples conditional U-statistics. We provide a  theoretical study of these estimators in Appendix \ref{sec:conv_knn}. We assume that the space of patients characteristics $\mathcal{X}$ is the normed  vector space $(\R^d,\norm{\cdot})$.

We consider  the function
$\sigma \colon \mathcal{Y} \times \mathcal{Y} \to \{-1;0;1\}$ defined by
\begin{equation*}
\begin{cases}
\sigma(a,b)=-1 \; \textnormal{if} \; b \prec a     \\
\sigma(a,b)=0 \; \textnormal{if} \;  a \bowtie b  \\
\sigma(a,b)=1 \; \textnormal{if} \; b \succ a   \\
\end{cases}
\end{equation*}
and observe that
\begin{align*}
    \Delta_{(r_0,r_1)}(X_i,U_j) &=\Esp(\sigma(Y_i(0),V_j(1)) \vert X_i,U_j)\\
                            &=\Esp(\sigma(Y_i,V_j) \vert X_i,U_j).
\end{align*}

We rely on nearest neighbors methods to define two-samples conditional U-statistics to estimate $\Delta_{(r_0,r_1)}$. 
Let $1 \leq c_m \leq m$ and $1 \leq e_n \leq n$ and set $N=m+n$. We estimate $\Delta_{(r_1,r_0)}$ via
\begin{equation}\label{eq:Delta_estimate}
    \widehat{\Delta}^N_{(r_0,r_1)}(x,u)=\sum_{i,j} W^N_{mi,nj}(x,u) \, \sigma(Y_i,V_j)
\end{equation}
where the weights $W^N_{mi,nj}$ are defined by
\begin{equation*}
W^N_{mi,nj}(x,u)=\begin{cases}
      \frac{1}{c_m \, e_n} \; \textnormal{if $X_i$ is among the $c_m$-NN of $x$ and $U_j$ is among the $e_n$-NN of $u$}\\
     0 \; \textnormal{otherwise}.
     \end{cases}
\end{equation*}
This yields
\begin{equation*}
    \widehat{\IPB}^N_{(r_0,r_1)}(x)=\widehat{\Delta}^N_{(r_0,r_1)}(x,x)
\end{equation*}
and we estimate the optimal rule by
\begin{equation*}
    \widehat{r}_N^{\,\opt}(x)=\mathds{1}_{\{\widehat{\IPB}^N_{(r_0,r_1)}>0\}}(x).
\end{equation*}

We prove in Corollary \ref{cor:wuc} that the estimator \eqref{eq:Delta_estimate} is universally weakly pointwise consistent under appropriate conditions (see Section \ref{sec:theoritical_result} and Online Appendix \ref{sec:conv_knn}).
Methods based on nearest neighbors often perform poorly when dealing with a large number of categorical covariates. This stems from the curse of dimensionality. Furthermore, choosing a suitable encoding scheme for categorical variables and defining a meaningful distance metric for mixed data types presents a significant challenge. Therefore, in the next section, we propose a classification-based approach to estimate the optimal ITR.

\subsection{Classification-based approach}\label{sec:ITRclassif}

We consider data from patients in the control group $(X_i, Y_i)_{1 \leq i \leq m}$ and from patients in the experimental group $(U_j, V_j)_{1 \leq j \leq n}$, and choose a classifier. Our methodology is the following one:

\begin{enumerate}[(1)]
    \item Construct a dataset consisting of all pairwise combinations of the form $(X_i, U_j, \sigma(Y_i, V_j))$ for $1 \leq i \leq m, 1 \leq j \leq n$,
\item Train the classifier $\rho(x, u)$ using the constructed dataset to predict the label $\sigma(Y, V)$ given the input pair $(X, U)$,
\item For a given patient with characteristic $x$, estimate $\IPB_{(r_0, r_1)}(x)$ as follows:
\begin{enumerate}
    \item Predict the probability $p_1(x)$ that $\rho(x, x) = 1$, and the probability $p_{-1}(x)$ that $\rho(x, x) = -1$,
    \item Define $\widehat{\IPB}^N_{(r_0, r_1)}(x) = p_1(x) - p_{-1}(x)$.
\end{enumerate}

\item Estimate the optimal rule using:
   \[
   \widehat{r}_N^{\,\opt}(x) = \mathds{1}_{\{\widehat{\IPB}^N_{(r_0, r_1)} > 0\}}(x).
   \]
\end{enumerate}
While this method is conceptually simple, it has a quadratic memory footprint, making it challenging to apply to large datasets. Additionally, many learning algorithms assume independent and identically distributed samples, which this approach does not inherently satisfy. To address these limitations, we propose a bagging-based method and prove that it is pointwise consistent, provided the base learner is pointwise consistent.

\section{Theoretical results}\label{sec:theoritical_result}

\subsection{Pointwise consistency of two-sample conditional \texorpdfstring{$U$}{U}-statistics}

We provide sufficient conditions to ensure that the estimator defined in Equation \eqref{eq:Delta_estimate} is pointwise consistent. We have the following result.

\begin{theorem}\label{thm:cv_knn_Delta}
Suppose that $\Delta_{(r_0,r_1)}$ is continuous on the support of $\mu \otimes \mu$. Set $N=n+m$ and assume that the sequences of integers ${c}=\{c_m\}$ and $e=\{e_n\}$ are such that

 \begin{enumerate}[(i)]
        \item $\dfrac{n}{m} \underset{N \to \infty}{\longrightarrow} \lambda \neq 0$, 
        \item $c_m \underset{N \to \infty}{\longrightarrow} \infty$ and $e_n \underset{N \to \infty}{\longrightarrow} \infty$,
        \item $\dfrac{c_m}{N} \underset{N \to \infty}{\longrightarrow} 0$ and $\dfrac{e_n}{N} \underset{N \to \infty}{\longrightarrow} 0$.
    \end{enumerate}
Then the two-samples conditionnal $U$-statistics $\widehat{\Delta}^N_{(r_0,r_1)}$ satisfies
\begin{equation*}
\Esp|\widehat{\Delta}^N_{(r_0,r_1)}(x,u)-\Delta_{(r_0,r_1)}(x,u) |^2 \to 0 \; \textnormal{for all} \;(x,u)\; \textnormal{in} \supp(\mu \otimes \mu). 
\end{equation*}
In particular, for all $(x,u) \in \supp(\mu \otimes \mu)$,
\begin{equation*}
    \widehat{\Delta}^N_{(r_0,r_1)}(x,u) \to \Delta_{(r_0,r_1)}(x,u) \; \textnormal{in probability}.
\end{equation*}
\end{theorem}

\begin{corollary}\label{cor:wuc}
\begin{enumerate}[(i)]
    \item For all $x \in \supp{\mu}$, $\widehat{\IPB}^N_{(r_0,r_1)}(x) \to \IPB_{(r_0,r_1)}(x)$ in probability  when $N \to \infty$.
    \item For all $x \in \supp(\mu)$, $\Prob(\widehat{r}^{\,\opt}(x) \neq r^{\,\opt}(x) ) \to 0$ when $N \to \infty$ .
\end{enumerate}
\end{corollary}

\subsection{Pointwise consistency of bagging}

The classification-based approach exhibit a substantial memory footprint, scaling quadratically with the data size that is the total number of patients. To address this, we propose a bagging procedure and demonstrate its consistency, provided that the base learner is consistent. Our strategy is inspired by \cite[Section 5]{Biau08}.

Given an iid sample formed of pairs of the form $(X_i,Y_i,U_i,V_i)_{1 \leq i \leq K}$, it induces an iid sample which elements are of the form $(X_i,U_i,\sigma(Y_i,V_i))_{1 \leq i \leq K}$. We denote this sample by $\mathcal{D}^\prime_K$. Using it, it is possible to learn
\begin{equation}\label{eq:target}
    \Delta(x,u)=\Esp(\sigma(Y,V) \mid X=x\,,U=u)
\end{equation}
using any appropriate learner $g_K(X,U)$.

However, we are not directly provided, with a dataset of the form $\mathcal{D}^\prime_K$. We are initially provided with the dataset $D_{m,n}:=C_m \sqcup E_n$, with
\begin{align*}
    C_m := & (X_1,Y_1), \ldots , (X_m,Y_m),\\
    E_n : = & (U_1,V_1), \ldots , (U_n,V_n).
\end{align*}
A first approach to learning \eqref{eq:target} is to consider the set of $mn$ pairs $\{(X_i,U_j,\sigma(Y_i,V_j))\}_{1 \leq i \leq m, 1 \leq j \leq n}$. However, as noted, this method has a quadratic memory footprint in the total number of patients $N = m+n$. Moreover, these pairs are not independent, which theoretically limits the applicability of most learning algorithms. A simpler alternative is to construct $k = \min(m,n)$ iid pairs of the form $(X_i,Y_i,U_i,V_i)_{1 \leq i \leq k}$, yielding an iid sample of $k$ variables $(X_i,U_i,\sigma(Y_i,V_i))$. This dataset, denoted $\mathcal{D}^\prime_k$, allows learning $\Delta(x,u)$ but is inefficient, as it utilizes only $\min(m,n)$ of the $mn$ available pairs. To address these limitations, we propose a bagging method based on randomized regressors. A randomized regressor incorporates a random variable $Z$ in its decision process.

 We introduce several notations in order to define our bagging procedure. We write $\mathcal{I}(\ell,m)$ for the set of injections from $\llbracket 1 ; \ell \rrbracket$ to $\llbracket 1 ; m \rrbracket$ and $\mathcal{I}_i(\ell,n)$ for the set of increasing injections from $\llbracket 1 ; \ell \rrbracket$ to $\llbracket 1 ; n \rrbracket$. Given $\alpha \in \mathcal{I}(\ell,m)$ and $ \beta \in \mathcal{I}_i(\ell,n)$, we extract the following sample from $\mathcal{D}_{m,n}$
\begin{equation*}
\mathcal{D}^\prime_{\alpha,\beta} \colon = \{(X_{\alpha(i)},U_{\beta(i)},\sigma(Y_{\alpha(i)},V_{\beta(i)}))\}_{1 \leq i \leq \ell}.
\end{equation*}
By construction, the sample $\mathcal{D}^\prime_{\alpha,\beta}$ is iid. We now introduce the random variable $Z$:
\begin{enumerate}[1.]
\item Sample an integer $\ell \in  \llbracket 1 ; k \rrbracket$ according to a binomial law with parameters $q_N \in [0,1]$ and $k=\min(m,n)$,
\item Sample uniformly an element $\alpha \in \mathcal{I}(\ell,m)$ and $\beta \in \mathcal{I}_i(\ell,n)$.
\end{enumerate}
Then, a realisation of $Z$ is the pair $(\alpha,\beta) \in \mathcal{I}(\ell,m) \times \mathcal{I}_i(\ell,m)$. We obtain the random training set $\mathcal{D}_{m,n}^\prime(Z)$, defined for each realisation $(\alpha,\beta)$ of $Z$ by $\mathcal{D}^\prime_{\alpha,\beta}$. The size of the dataset $\mathcal{D}_{m,n}^\prime(Z)$ is a binomial random variable $L$ of parameter $q_k$ and $k$.

Given a sequence of estimators $\{g_K(x,u,\mathcal{D}^\prime_K)\}_{K \in \N}$ of $\Delta(x,u)$, We can associate to it a sequence of randomized estimators $g_L(x,u,\mathcal{D}_{m,n}^\prime(Z))$.

Consider an iid sample $Z_1,\ldots,Z_b$ such that $Z_i \sim Z$. We form the $b$ randomized regressors 
\begin{equation*}
g_{L_1}(\cdot,\mathcal{D}^\prime_{m,n}(Z_1)),\ldots,g_{L_b}(\cdot,\mathcal{D}^\prime_{m,n}(Z_b)). 
\end{equation*}
We then consider the averaging regressor associated to $b$ draws of the randomizing variable $Z$, that is
\begin{align} \label{eq:average_clas}
    g_{m,n}^{(b)}(x,u,Z^b,\mathcal{D}_{m,n})=\frac{1}{b}\sum_{v=1}^b g_{m,n}(x,u,Z_v,\mathcal{D}_{m,n}).
\end{align}
The following result is an analogue of \cite[Theorem 6]{Biau08}.
\begin{theorem}\label{thm:bagging} Let $\{g_K\}$ be a sequence of regressors. Consider the averaging regressor $g_{m,n}^{(b)}(x,u,Z^b,\mathcal{D}_{m,n})$ defined above.  Assume that 
    \begin{enumerate}[(i)]
        \item $g_K(x,u, \mathcal{D}^\prime_K)$ is pointwise consistent in $(x,u)$ for the distribution of $(X,U,\sigma(Y,V))$,
        \item  $k q_{k} \to \infty$ as $k \to \infty$ where $k=\min(m,n)$.
    \end{enumerate}
   Then the  regressor $g_{m,n}^{(b)}(x,u,Z^b,\mathcal{D}_{m,n})$ is pointwise consistent in $(x,u)$.
\end{theorem}

\section{Simulation study}\label{sec:simul}

\subsection{Simulation settings and analyses}

We conducted a simulation study to evaluate the finite-sample properties of the classification methodology. Notably, we did not assess the estimator provided by formula \eqref{eq:Delta_estimate}, as it relies on nearest neighbors, requiring a meaningful distance metric on the space of patient characteristics. Selecting an appropriate metric is challenging when dealing with (1) multiple categorical variables and (2) a large number of covariates. Although theoretically appealing, this estimator is difficult to apply in practice for these reasons. Instead, we illustrate the methods described in Section \ref{sec:ITRclassif}, using a random forest classifier as the base learner.

\paragraph{Simulation scenarios:} We examined two distinct scenarios. The simulated population included eight correlated covariates: one normal, three log-normal, and four Bernoulli. The outcomes and scores varied between scenarios:
\begin{itemize}
    \item \textbf{Scenario 1:} Outcomes consist of two binary random variables, with the score defined in Table \ref{tab:score}.
    \item \textbf{Scenario 2:} Outcomes comprise one binary variable and one integer-valued variable (ranging from 0 to 25), with the score defined in Table \ref{tab:score_multi} using $\delta=3$.
\end{itemize}
In both cases, we computed the oracle $\Delta_{(r_0,r_1)}$ and the corresponding optimal individualized treatment rule (ITR). Details of the data generation mechanism are provided in Appendix \ref{sec:datgen}.

\paragraph{Data Generation:} For each scenario, we generated two datasets of 2 million individuals each: one for training the models and the other for evaluation. From the large training dataset, we created 1000 datasets of size 400 hundreds and 1000 datasets of size 2000. Each of these datasets was split into two equal-sized folds, one serving as the control group and the other as the experimental group. 

\paragraph{Model training and hyperparameter tuning:} We formed all possible pairs $(X, U, \sigma(Y(0), V(1)))$, where $(X, Y(0), Y(1))$ belong to a patient in the control group, and $(U, V(0), V(1))$ belong to a patient in the experimental group. We then trained a random forest classifier to predict $\sigma(Y(0), V(1))$ from the data of $(X,U)$. The hyperparameters of the random forest classifier (number of trees and features per tree) were tuned using 5-fold cross-validation over fifty datasets. Once determined, these parameters were fixed for subsequent analyses. The optimal rule was estimated by $\widehat{r^{\opt}}=\mathds{1}_{\{\widehat{\IPB}_{(r_0,r_1)}>0\}}$.

\paragraph{Performance metrics:}  In order to assess the quality of the estimation of the $\IPB$ and of the learned ITR, we use the following metrics. We compute them for each training set and report the result averaged over 1000 iterations.

\begin{itemize}
    \item We assess the quality of the prediction of $\IPB_{(r_0,r_1)}$ by computing its root mean square error (RMSE) averaged over one thousand datasets. We also compute the bias of the estimation of the $\AIPB_{(r_0,r^\opt)}$, defined in equation \eqref{eq:AIPB}, which compare the performance of two ITRs.
    
    \item The estimated ITR can be considered as a plug-in classifier whose decision function is  $\widehat{\IPB} _{(r_0,r_1)}$. Hence, to estimate the overall performance of the learned ITRs, we compute its AUC which is the area under the Receiver Operating Characteristic (ROC) curve and report the averaged AUC of the ITRs learned across 1000 datasets. For the same reason, we report the average of the Matthews correlation coefficient of the ITRs learned accross 1000 datasets produced by the same generating mechanism.
    
    \item We also report the averaged specificity and sensitivity of the estimated ITRs, which allows us to assess the ability of the ITRs to properly identify patients who should receive the control treatment (specificity) and those who should receive the experimental treatment (sensitivity).

    \item We provide the averaged confusion matrices of the ITRs (Table \ref{tab:confusion_matrix}) and the averaged calibration curves of the ITRs for each of the scenarios (Figure \ref{fig:grid}).
\end{itemize}

\subsection{Results}

We report the following  measures:
\begin{table}[H]
\centering
\caption{Performance metrics of the estimation of the $\IPB$ and the estimated optimal ITR averaged across 1000 simulation iterations in two scenarios and two sample sizes.}
\begin{tabular}{cccccc}
\hline
            & \multicolumn{2}{c}{Scenario 1 } & \multicolumn{2}{c}{Scenario 2}  \\
            &\multicolumn{2}{c}{$\AIPB_{(r_0,r^\opt)}=0.18$} & \multicolumn{2}{c}{$\AIPB_{(r_0,r^\opt)}=0.20$}\\
            \hline
Trial of size & 400 & 2000 & 400 & 2000  \\ 
\hline RMSE($\IPB_{(r_0,r_1)}$)   (S.E)   & 0.36 (0.02)   & 0.27 (0.01)  &  0.38 (0.02)   &  0.29 (0.01)    \\
$\widehat{\AIPB}_{(r_0,r^{\opt})}-\AIPB_{(r_0,r^{\opt})}$ (S.E) & -0.03 (0.03)   &     -0.02 (0.01)       &  -0.03 (0.03)         &    -0.02  (0.01)       \\
AUC   (S.E)  &    0.90  (0.02)    &      0.94 (0.01)    &     0.91  (0.02)    &  0.95 (0.01)   \\
Matthew correlation coefficient  (S.E)  &  0.62 (0.05)   &  0.70 (0.02)   &  0.64 (0.04)  &     0.74 (0.02)  \\
Specificity (S.E) & 0.81 (0.03) & 0.86 (0.02) & 0.84 (0.03) & 0.88 (0.02)\\
Sensitivity (S.E) & 0.73 (0.07) & 0.80 (0.03) &0.77 (0.06) & 0.83 (0.03)\\ \hline
\end{tabular}
\end{table}

\begin{table}[H]
\caption{Confusion Matrix for the different scenarios and trial sizes}
    \centering
    \setlength{\tabcolsep}{8pt} 
    \renewcommand{\arraystretch}{1.5} 
    \begin{tabular}{ccccc}
        \hline
        && Treatment &Trial of size 400 &Trial of size 2000 \\
        \hline
        Scenario 1 & \multirow{ 2}{*}{\rotatebox{90}{Best treatment  }} &
        \begin{tabular}{c}
             0 \\
             1
        \end{tabular} &
        \begin{tabular}{cc}
             49.2\% (2.3) & 7.4\% (2.3) \\
            
            11.5\% (2.9) & 31.9 (2.9)\%
        \end{tabular} &
        \begin{tabular}{cc}
            50.7\% (1.1) & 5.9\% (1.1) \\
            8.6\% (1.4) & 34.8\% (1.4)
        \end{tabular} \\
        
        Scenario 2 & &
        \begin{tabular}{c}
             0 \\
             1
        \end{tabular} &
        \begin{tabular}{cc}
            50.1\% (1.9) & 7.4\% (1.9) \\
            10.1\% (2.6) & 32.5\% (2.6)
        \end{tabular} &
        \begin{tabular}{cc}
            51.8\% (0.9) & 5.6\% (0.9)\\
            7.2\% (1.1) & 35.4\% (1.1)
        \end{tabular} \\
        \hline
        &&Treatment &
        \begin{tabular}{cc}
        0 & 1
        \end{tabular}&
        \begin{tabular}{cc}
        0 & 1
        \end{tabular}
        \\
        \hline&&&\multicolumn{2}{c}{Recommended treatment}\\
    \end{tabular}
    
    \label{tab:confusion_matrix}
\end{table}

\begin{figure}[ht]
    \centering
    \begin{minipage}{0.45\textwidth}
        \centering
        \includegraphics[width=\textwidth]{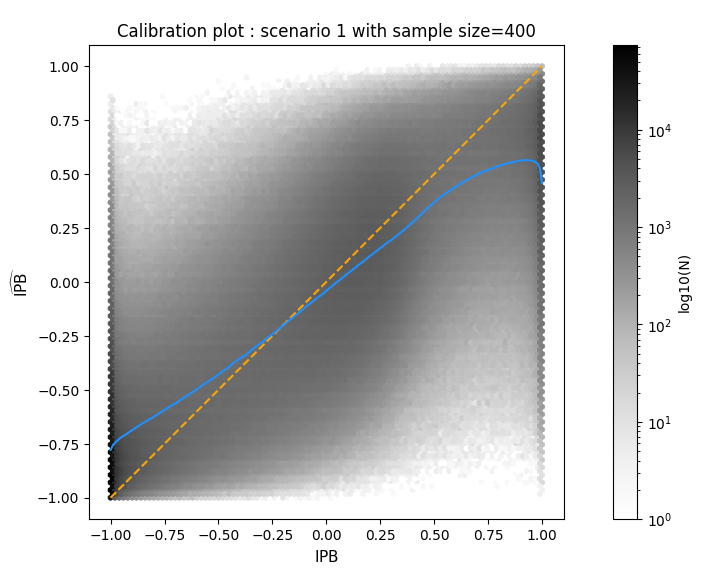}
    \end{minipage}
    \hfill
    \begin{minipage}{0.45\textwidth}
        \centering
        \includegraphics[width=\textwidth]{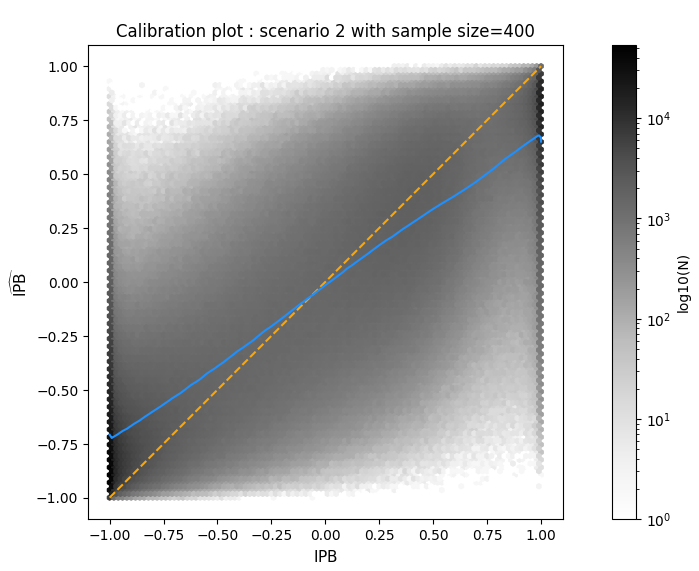}
    \end{minipage}
    
    \vspace{1em} 
    
    \begin{minipage}{0.45\textwidth}
        \centering
        \includegraphics[width=\textwidth]{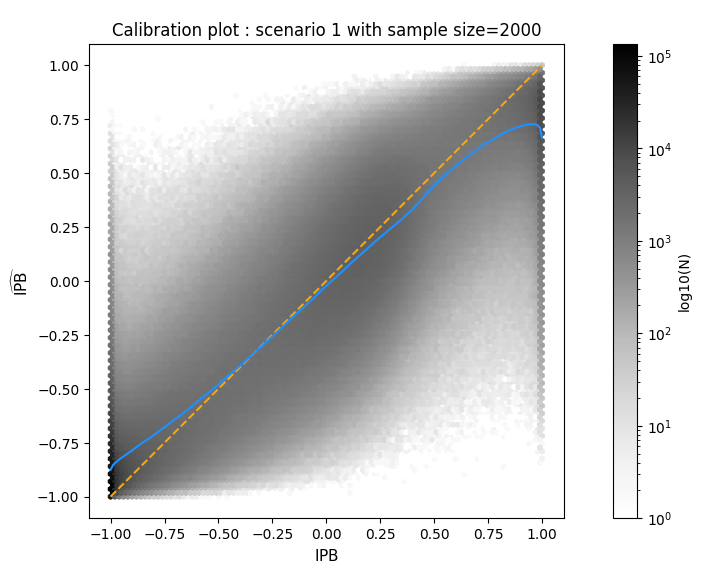}
    \end{minipage}
    \hfill
    \begin{minipage}{0.45\textwidth}
        \centering
        \includegraphics[width=\textwidth]{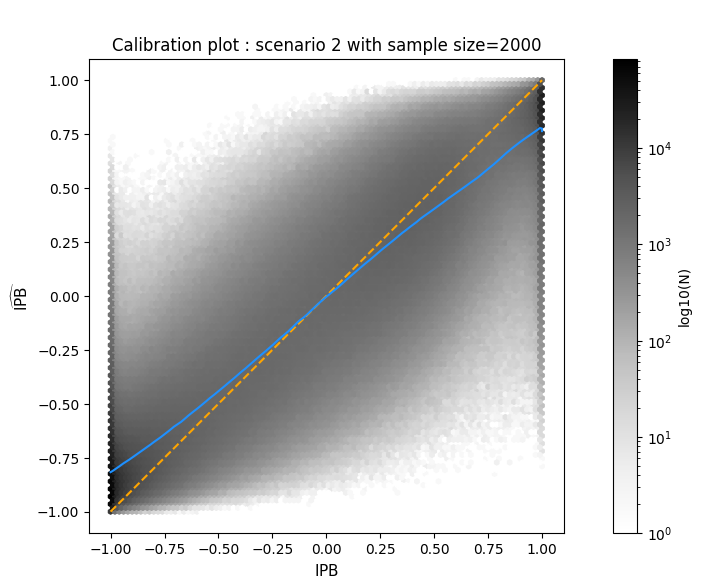}
    \end{minipage}
    
    \caption{Calibration plot of the estimated $\widehat{\IPB}$ versus the oracle $\IPB$ for various scenarios and sample sizes. The calibration curve is shown in blue, the diagonal in orange, and the density of the tuples $(\IPB, \widehat{\IPB})$ is represented using a grayscale scale. For sake of presentation random samples of 100\,000 points are drawn.}
    \label{fig:grid}
\end{figure}

We notice, through the various metrics reported, that the estimation improves when the sample size increase.

\section{Illustration: IST-3 trial}\label{sec:data}

We illustrate our method to estimate an optimal ITR for the thrombolysis of patients with ischemic stroke within 6 hours of symptom onset.
For that purpose, we use the data from the third International Stroke Trial (IST-3) an international, multicenter, randomized controlled trial that included 3\,035 participants recruited from 156 hospitals in 12 countries \cite{ist3}. Participants had ischemic stroke and could start rt-PA within 6 hours of symptom onset; they were randomized between intravenous thrombolysis with recombinant tissue plasminogen activator (rt-PA) or control (same stroke care without rt-PA).

The primary outcome of this trial was the proportion of participants alive and independent at 6 months, as measured using the Oxford Handicap Score (OHS) \cite{Bamford1990}. Patient with an OHS of 0,1 or 2 were considered independent. Several secondary outcomes were reported, among which health-related quality of life at 6 months measured through the EuroQol scale \cite{Herdman2011}.

The IST-3 trial failed to show an overall improvement in the proportion of patient alive and independent at 6 months with rt-PA, but suggested an improvement in functional outcomes at 6 months despite early risks \cite{ist3}. Authors concluded that further studies should be conducted in selected
patients more than 4.5 hours after stroke. Further, the trial did not reach its originally planned sample size of 6\,000 participants, and was therefore underpowered to detect the targeted 3\% absolute difference in the primary outcome. The absolute risk difference in primary outcome was an improvement by 1.4\% (95\% confidence interval $-$0.2 to 4.8) with rt-PA. A secondary analysis showed a significant favorable shift in the distribution of the OHS score. On average, patients treated in the rt-PA group survived with less disability.

In this reanalysis, we considered hierarchical outcomes, with the OHS at six months as the higher priority outcome, and the EuroQol at six months as the lower priority outcome (Table \ref{tab:2levelscore}). For comparison, we also used our method with a single outcome, OHS at six months (online Appendix~\ref{sec:ohs6}).

\begin{table}[htb]
	\centering
	\caption{Scoring $\sigma_{ij}$ of a pair $(i,j)$ of individuals taken from the control ($C$) and experimental ($E$) groups, respectively, with  outcomes $(Y^{1}_i,V^{1}_j)$ and $(Y^{2}_i,V^{2}_j)$ with threshold $\delta=5$.}
     \label{tab:2levelscore}
	\begin{tabular}{lcclccc}
		\hline
		\multicolumn{3}{l}{Outcome with higher priority:} & \multicolumn{3}{l}{Outcome with lower priority:}  & \\ 
        \multicolumn{3}{l}{OHS at 6 months} & \multicolumn{3}{l}{EuroQol score at 6  months}  & \\ \cline{1-6}
		$(Y^{1}_i,V^{1}_j)$ & Ordering & Label & $(Y^{2}_i,V^{2}_j)$ & Ordering & Label & $\sigma_{ij}$ \\ \hline
		$V_j^1-Y_i^1>0$ &  $V^{1}_j \succ Y^{1}_i$ & Favorable & Any & --- & Not considered & $+ 1$ \\
		$V_j^1-Y_i^1<0$ & $V^{1}_j \prec Y^{1}_i$ & Unfavorable & Any & --- & Not considered & $-1$\\
		$V_j^1-Y_i^1=0$ & $V^{1}_j \bowtie Y^{1}_i$ & Tie/neutral & $Y^{2}_i-V^{2}_j<-5$ & $V^{2}_j \succ Y^{2}_i$ & Favorable & $+ 1$\\
		$V_j^1-Y_i^1=0$ & $V^{1}_j \bowtie Y^{1}_i$ & Tie/neutral &  $Y^{2}_i-V^{2}_j>5$ & $V^{2}_j \prec Y^{2}_i$ & Unfavorable & $-1$\\
		$V_j^1-Y_i^1=0$ & $V^{1}_j \bowtie Y^{1}_i$ & Tie/neutral & $ \vert Y^{2}_i-V^{2}_j \vert \leq 5$ & $V^{2}_j \bowtie Y^{2}_i$ & Tie/neutral & $0$\\
		\hline
	\end{tabular}
\end{table}

The ITR derivation was based on 61 covariates collected before randomization, 47 of which being categorical. The list of the covariates included in the model is available in the online Appendix \ref{sec:covar}.

To estimate the ITRs, we used the methods described in Section \ref{sec:ITRclassif}. We first imputed the missing data using multivariate imputation by chained equations \cite{vanBuuren11}. For the sake of illustration, we only used one imputed dataset, while in real application the process should be repeated several times.  We used a random forest classifier as based learner. The hyper-parameters (number of trees and number of features per tree) were chosen via a 5-fold cross-validation. 

We derived the optimal individual treatment rule (ITR) using the entire dataset. The characteristics of patients, grouped by the recommended treatment, are summarized in Table \ref{tab:char_recom1}. We write $r^{\opt}_h$ for the optimal rule associated with the score of Table \ref{tab:2levelscore}. 

We estimated the AIPB for the rule $r_1$, which assigns rt-PA to all patients, using a G-computation-type approach. To mitigate overfitting, we employed a 3-fold cross-fitting procedure: the $\IPB$ was learned on two folds, and the AIPB was estimated on the remaining fold. This procedure was repeated three times, each time rotating the role of the held-out fold. Finally, the estimates from each fold were averaged. The standard error was computed using 200 bootstrap iterations.  We also estimated the AIPB of $r^{\opt}_h$ relative to both $r_1$ and $r_0$ (no treatment). This corresponds to evaluating the rule learned on the entire dataset, along with the associated standard errors conditional to the dataset. To mitigate overfitting, we employed a leave-4-out procedure. This approach ensures that the learned rules are similar across folds, allowing us to reasonably assume that the estimated $\IPB$ values are also consistent from one fold to another.

Overall, results showed a modest benefit of rt-PA, with a proportion in favor of treatment of 4.6\% (Table \ref{tab:result_ist3}). In addition, following a personalized treatment strategy further improved the outcome, both compared to giving rt-PA to everyone or to nobody (Table \ref{tab:result_ist3}).

\begin{table}[htb]
	\centering
	\caption{Estimate of the AIPB for the rule treating everyone with rt-PA and for the estimated optimal rules together with the standard errors. The last two standard errors are conditional to the training set.}\label{tab:result_ist3}
	\begin{tabular}{ccc}
    \hline
    \rule{0pt}{2.8ex} 
        Treatment rules & Mean & Standard error \\ \hline 
		$\widehat{\AIPB}_{(r_0,r_1)}$ & 0.046 & (0.019)\\
        $\widehat{\AIPB}_{(r_1,\hat{r}_h^{\opt})}$ & 0.040 & (0.001) \\
        $\widehat{\AIPB}_{(r_0,\hat{r}_h^{\opt})}$ & 0.089 & (0.002)\vspace{0.1cm}\\
    \hline
	\end{tabular}
\end{table}
A closer look at Table \ref{tab:char_recom1} shows that rt-PA is more often recommended for patients with mild-to-moderate stoke (NIHSS 6–15) and who have better general condition. Interestingly, rt-PA was not less recommended for patients with longer delays from symptoms onset. However, there is no evident simple rule to decide upon treatment, thereby illustrating the interest of our approach given the overall advantage of a personalized treatment strategy. 
{\footnotesize
\begin{longtable}{lll}
\caption{ Characteristics of patients according to the treatment recommended by $\hat{r}_h^{\opt}$. NIHSS=National Institutes of Health Stroke Scale. TACI=total anterior circulation infarct. PACI=partial anterior circulation infarct. LACI=lacunar infarct. POCI=posterior circulation infarct.}
\label{tab:char_recom1}\\
\hline
& \textbf{rt-PA} & \textbf{No rt-PA} \\
& \textbf{recommended} & \textbf{recommended} \\
& \textbf{(n=1670)} & \textbf{(n=1365)} \\
\hline
\endfirsthead

\hline
& \textbf{rt-PA} & \textbf{No rt-PA} \\
& \textbf{recommended} & \textbf{recommended} \\
& \textbf{(n=1670)} & \textbf{(n=1365)} \\
\hline
\endhead
\multicolumn{3}{l}{\textbf{Baseline variables}} \\
\hline
\multicolumn{3}{l}{\textbf{Age (years)}} \\
18--50 & 4.7\% &  3.6\% \\
51--60 & 6.5\% & 6.9\%  \\
61--70 & 12.3\% & 11.7\%\\
71--80 & 22.3\% & 25.8\%  \\
81--90 & 47.7\% & 44.7\% \\
$>$90 & 6.6\% & 7.3\% \\
\hline
\multicolumn{3}{l}{\textbf{Sex}} \\
Female & 52.9\% & 50.3\% \\
\hline
\multicolumn{3}{l}{\textbf{NIHSS}} \\
0--5 & 20.3\% & 19.9\% \\
6--10 & 28.7\% & 27.3\% \\
11--15 & 19.6\% & 20.0\% \\
16--20 & 17.8\% & 18.0\% \\
$>$20 &  13.5\% & 14.7\% \\
\hline
\multicolumn{3}{l}{\textbf{Delay in randomisation}} \\
0--3.0 h & 21.0\% & 19.7\% \\
3.0--4.5 h & 36.9\% & 39.0\% \\
4.5--6.0 h & 38.0\% & 37.6\% \\
$>$6.0 h & 4.1\% & 3.7\% \\
\hline
Atrial fibrillation & 29.9\% & 30.4\% \\
\hline
\multicolumn{3}{l}{\textbf{Systolic blood pressure}} \\
$\leq$143 mm Hg & 34.9\% & 31.8\% \\
144--164 mm Hg & 31.6\% & 33.0\% \\
$\geq$165 mm Hg & 33.5\% & 35.2\% \\
\hline
\multicolumn{3}{l}{\textbf{Diastolic blood pressure}} \\
$\leq$74 mm Hg & 34.1\% & 32.8\% \\
75--89 mm Hg & 33.8\% & 34.0\% \\
$\geq$90 mm Hg & 32.2\% & 33.2\% \\
\hline
\multicolumn{3}{l}{\textbf{Blood glucose}} \\
$\leq$5 mmol/L & 20.8\% & 18.5\% \\
6--7 mmol/L & 46.8\% & 48.8\% \\
$\geq$8 mmol/L & 32.5\% & 32.7\% \\
\hline
Treatment with antiplatelet drugs in previous 48 h & 51.0\% & 52.0\% \\
\hline
\multicolumn{3}{l}{\textbf{Predicted probability of poor outcome at 6 months}} \\
$<$40\% & 54.3\% & 54.6\% \\
40--50\% & 11.4\% & 9.9\% \\
50--75\% & 23.1\% & 26.1\% \\
$\geq$75\% & 11.2\% & 9.5\% \\
\hline
\multicolumn{3}{l}{\textbf{Stroke clinical syndrome}} \\
TACI & 42.7\% & 43.4\% \\
PACI & 37.7\% & 37.8\% \\
LACI & 11.6\% & 10.2\% \\
POCI & 7.8\% & 8.4\% \\
Other & 0.2\% &0.1\% \\
\hline
\multicolumn{3}{l}{\textbf{Baseline variables collected from prerandomisation scan}} \\
\multicolumn{3}{l}{\textbf{Expert reader's assessment of acute ischaemic change}} \\
Scan completely normal & 9.2\% & 9.2\% \\
Scan not normal but no sign of acute ischaemic change & 50.4\% & 50.6\% \\
Signs of acute ischaemic change & 40.4\% & 40.1\% \\
\hline
\end{longtable}
}

\section{Discussion} \label{sec:discussion}

This paper introduces an optimality criterion for individualized treatment rules (ITRs), based on net treatment benefit, and estimation methods using Generalized Pairwise Comparison \cite{Buyse:2010aa}. We present two estimators: a nearest-neighbor approach and a meta-learner similar to the S-learner. Our framework accommodates prioritized outcomes, enabling multi-outcome treatment personalization beyond simple benefit-risk assessments. While prioritized outcomes have been explored in treatment personalization \cite{Duke2024}, our work uniquely applies a pairwise comparison perspective.

Furthermore, our approach is applicable to single-outcome settings with individualized net benefit, maintaining a computational complexity comparable to the S-learner. While it accurately estimates the Conditional Average Treatment Effect (CATE) in binary outcomes, it offers neither a distinct advantage nor a theoretical disadvantage. However, its memory footprint is substantial, scaling quadratically with N=n+m, the total number of patients, limiting its practicality for large datasets. This limitation is mitigated by our proposed bagging approach. Conversely, for continuous outcomes, the situation differs, and this method may be of interest. Specifically, it can filter out small, clinically insignificant differences (e.g., below the minimal detectable change or minimal clinically important difference) while retaining a quantitative measure of the outcome. This minimizes information loss compared to binarization.

As mentioned earlier, the classification-based approach is not very efficient due to the large number of pairs involved, which leads to a significant memory footprint. Furthermore, the framework is mathematically complex, posing challenges even in defining the relevant objects and making it difficult to establish the properties of the estimators, mainly due to the non-independence of pairs. One possible way to address this issue is through Generalized Estimating Equations (GEE), and specifically the Probabilistic Index Model (PIM) \cite{Thas2012}, which can handle sparsely correlated variables in a semi-parametric framework. However, we choose to situate our approach in a non-parametric setting. Additionally, we have not studied confidence intervals in this work, but for conditionally U-statistic-based estimators, asymptotic normality could be established using standard arguments. Nevertheless, inference remains an open research question for the classification-based approach.

This work has been developed for both binary and continuous outcomes, but handling censored data would be valuable, although it presents significant difficulties \cite{Deltuvaite2023,Dong2020}. This complexity is already evident in the standard Generalized Pairwise Comparison framework. Further exploration of this aspect would be useful. While our approach has been applied in the context of randomized trials, it can be extended quite naturally to observational settings. Additionally, although we focused on the net benefit criterion, our methodology is likely closely related to the win ratio. However, the latter is a ratio rather than a difference.


\section*{Acknowledgments}
Francois Petit acknowledge support
by the French Agence Nationale de la Recherche through the project reference ANR-22-CPJ1-0047-01. Raphaël Porcher acknowledges support by the French Agence Nationale de la Recherche as part of the “Investissements d’avenir” program, reference ANR-19-P3IA-0001 (PRAIRIE 3IA Institute). François Petit is grateful to Mathieu Even and Julie Josse for their insightful and stimulating discussions during the writing of this paper.

\bibliographystyle{sim}
\bibliography{GPC}

\clearpage
\appendix
\setcounter{page}{1}
\setcounter{table}{0}
\renewcommand\thetable{\thesection.\arabic{table}}
\setcounter{figure}{0}
\renewcommand\thefigure{\thesection.\arabic{figure}}

\section*{Appendix}

\section{Lebesgue differentiablity}

Lebesgue differentiability is fundamental to our approach, as it ensures the well-definedness of $\IPB_{(r_0,r_1)}$. We thus dedicate this section to a brief review of this concept.

We consider a metric space $(\mathcal{X},d)$ equipped with its Borel $\sigma$-algebra and $\mu$ a locally finite Borel regular measure. We denote by $B(x,r)$ the closed ball of center $x \in \mathcal{X}$ and radius $r$, i.e.,
\begin{equation*}
    B(x,r)=\{z \in \mathcal{X} | d(x,z) \leq r \}
\end{equation*}
where $r \in \R_{>0}$.
\begin{definition}
    Let $f\colon \mathcal{X} \to \R$ be a locally integrable function. We say that $x \in \mathcal{X}$ in the support of $\mu$ is a Lebesgue point of $f$ if
    \begin{equation*}
        \lim_{r \to 0^+} \frac{1}{\mu(B(x,r))} \int_{B(x,r)} |f(z)-f(x)| \, \mu(dz) = 0.
    \end{equation*}
\end{definition}
We remark that if $\mu$ is the probability distribution of a random variable $X$, then
\begin{equation*}
    \frac{1}{\mu(B(x,r))} \int_{B(x,r)} |f(z)-f(x)| \, \mu(dz)=\Esp(|f(X)-f(x)| \, | X \in B(x,r) ).
\end{equation*}
\begin{remark}
\begin{enumerate}[(i)]
    \item Let $f$ and $g$ be two locally integrable functions such that $f=g$ $\mu$-almost everywhere. Assume that $x$ is a Lebesgue point of $f$. Then
    \begin{equation*}
        \lim_{r \to 0^+} \frac{1}{\mu(B(x,r))} \int_{B(x,r)} |g(z)-f(x)| \, \mu(dz) = 0.
    \end{equation*}
    \item Consider the metric space $(\mathcal{X},d)$. Assume it is endowed  with a probability measure $\mu$. Let $f\colon \mathcal{X} \to \R$ be a locally integrable function. If $\mu(x)>0$, then $x$ is a Lebesgue point of $f$. This setting applies in particular to the situation where $\mathcal{X}$ is at most countable and $\mu$ is a discrete probability measure. In this setting, we can endow $\mathcal{X}$ with the discrete distance i.e., $d(x,y)=0$ if $x=y$ and  $d(x,y)=1$ otherwise .
    \item Assume that $f\colon \mathcal{X} \to \R$ is a continuous function. Then every point of $\mathcal{X}$ is a Lebesgue point of $f$.
\end{enumerate}
\end{remark}
One of the key results of Lebesgue differentiation theory is the Lebesgue differentiation theorem. This results holds in various level of generality. Here, we state it for Radon measure on a real finite dimensional normed vector space. We refer the reader to \cite[Theorem 2.9.8]{Fed} (together with  \cite[\textsection 2.8.9]{Fed} and \cite[Theorem 2.8.18]{Fed}) for a proof.

\begin{theorem}\label{thm:Lebesgue_diff}
    Let $(E,\norm{\cdot})$ be a real finite dimensional vector space endowed with a radon measure $\mu$. Let $f \colon E \to \R$ be a locally integrable function. Then $\mu$-almost every point of $E$ is a Lebesgue point of $f$.
\end{theorem}

\section{Individual pairwise comparison} \label{app:existence}
\subsection{Mathematical soundness}

\subsubsection{Technical assumptions} 

The aim of this section is to demonstrate that $\IPB$ is well-defined. At first glance, $\IPB$ appears to be defined as the restriction to a set of measure zero of an object-initially specified only up to a set of measure zero. However, this issue is merely superficial, and under suitable assumptions, $\IPB$ is indeed well-defined. Our approach involves constructing specific integrable functions that represent the conditional expectations $\Esp(\sigma(Y(i), V(j)) \vert X, U)$ for $i=0,1$ and $j=0,1$.
 
We assume that the space $\mathcal{X}$ of patient's characteristics is endowed with a distance $d_\mathcal{X}$. The product of two metric spaces is not canonically a metric space. Hence, the space of pairs of patient's characteristics $\mathcal{X} \times \mathcal{X}$ is not naturally endowed with a distance. Depending of situation we equip $\mathcal{X} \times \mathcal{X}$ with either the distance $d_{\mathcal{X} \times \mathcal{X}}((x,u),(x^\prime,u^\prime)) := \sqrt{d^2_\mathcal{X}(x,u)+d^2_\mathcal{X}(x^\prime,u^\prime)}$ or $d_{\mathcal{X} \times \mathcal{X}}((x,u),(x^\prime,u^\prime)):=\max(d_\mathcal{X}(x,u),d_\mathcal{X}(x^\prime,u^\prime))$

We denote the space of outcome by $(\mathcal{Y},d_\mathcal{Y})$ and assume it is a metric space endowed with a Borel measurable function $\sigma \colon \mathcal{Y} \times \mathcal{Y} \to \{-1;0;1\}$ such that $\sigma(y,v)=-\sigma(v,y)$. We write 
\begin{equation}\label{eq:score_function}
\begin{cases}
b \prec a \; \textnormal{if} \; \sigma(a,b)=-1    \\
a \bowtie b \; \textnormal{if} \;  \sigma(a,b)=0\\
b \succ a \; \textnormal{if} \; \sigma(a,b)=1 .\\
\end{cases}
\end{equation}
We further assume that  we are provided with two integrable iid populations, namely, $(X,Y(0),Y(1))$ and  $(U,V(0),V(1))$ which are treated respectively according to the rule $r \colon \mathcal{X} \to \{0;1\}$ and $s \colon \mathcal{X} \to \{0;1\}$ (which are assumed to be measurable).  The law of these populations is a Borel probability measure denoted $\mu$. Since $X$ and $U$ are iid the law of $(X,U)$, denoted $\mu_{X,U}$, is given by the product measure $\mu \otimes \mu$. Without loss of generality, we assume that $\mathcal{X}=\supp(\mu)$.

Here is our key hypothesis. We assume that for any pair $(x,u) \in \mathcal{X} \times \mathcal{X}$ and $i=0,1$, $j=0,1$ the following limits exist 
\begin{align} \label{hyp:funda}
 \lim_{\varepsilon \to 0} \Esp(\sigma(Y(i),V(j)) \vert (X,U) \in B((x,u),\varepsilon)) \tag{H.1}
\end{align}

and set
\begin{align}
    \Delta_{(r_0,r_1)}(x,u):=& \lim_{\varepsilon \to 0} \Esp(\sigma(Y(0),V(1)) \vert (X,U) \in B((x,u),\varepsilon)) \label{lim:delta01},\\
    \Delta_{(r_i,r_i)}(x,u):=&\lim_{\varepsilon \to 0} \Esp(\sigma(Y(i),V(i)) \vert (X,U) \in B((x,u),\varepsilon)). \label{lim:delta00}
\end{align}
 This implies in particular that the value $\Delta_{(r_0,r_1)}(x,x)$ is well-defined or in other word that $\IPB(x)$ is well-defined.
The Hypothesis \ref{hyp:funda} can be understood as claiming that the pairwise benefit of a pair is the average of the pairwise benefit of the neighboring pairs.

We need to verify that the functions $\Delta_{(r_i,r_j)}$ represent the conditional expectations $\Esp(\sigma(Y(i),V(j)) \vert X,U)$. This is the purpose of the following proposition.
\begin{proposition} \label{prop:espconradon}
Assume that
\begin{enumerate}[(i)]
    \item $\mathcal{X} \subset \R^d$,
    \item $d_{\mathcal{X} \times \mathcal{X}}$ is induced by a norm on $\R^d \times \R^d$,
    \item $\mu$ is a Radon measure.
\end{enumerate}
Then $\Delta_{(r_i,r_j)}$ is an integrable function on $\mathcal{X} \times \mathcal{X}$ such that $\Delta_{(r_i,r_j)}(X,U)=\Esp(\sigma(Y(i),V(j)) \vert X,U)$ almost surely.
\end{proposition}

\begin{proof}
Let $\phi_{(r_i,r_j)}$ be an integrable function such that i.e.,
$\Esp(\sigma(Y(i),V(j)) \vert X,U)=\phi_{(r_i,r_j)}(X,U)$ almost surely. Hence,
\begin{align*}
    \Esp(\sigma(Y(i),V(j)) \vert (X,U) \in B((x,u),\varepsilon))= \frac{1}{\mu \otimes \mu (B((x,u),\varepsilon))} \int_{B((x,u),\varepsilon)} \phi_{(r_i,r_j)}(a,b) \; d\mu(a) \otimes d\mu(b).
\end{align*}
Since $\R^d$ is second countable, $\mu \otimes \mu$ is again a Radon measure. Then, Theorem \ref{thm:Lebesgue_diff} implies that  
\begin{equation*}
\lim_{\varepsilon \to 0}    \frac{1}{\mu \otimes \mu (B((x,u),\varepsilon))} \int_{B((x,u),\varepsilon)} \phi_{(r_i,r_j)}(a,b) \; d\mu(a) \otimes d\mu(b) = \phi_{(r_i,r_j)}(x,u)
\end{equation*}
for $\mu \otimes \mu$-almost every $(x,u) \in \mathcal{X} \times \mathcal{X}$. Thus $\Delta_{(r_i,r_j)}=\phi_{(r_i,r_j)} \,\, \mu \otimes \mu-a.e.$ which conclude the proof.
\end{proof}

\begin{remark}
 The assumption \eqref{hyp:funda} together with the assumption $(i)$, $(ii)$  and $(iii)$ of Proposition \ref{prop:espconradon} can be replaced by the following hypothesis.
    \begin{enumerate}[($\textnormal{H}^\prime$.1)]
   \item There exist a $d_{\mathcal{X} \times \mathcal{X}}-\mu \otimes \mu $-Lebesgue differentiable function on $\mathcal{X} \times \mathcal{X}$ denoted $\Delta_{(r_0,r_1)}$ such that
    \begin{equation*}
\Esp(\sigma(Y(0),V(1)) \vert X,U)=\Delta_{(r_0,r_1)}(X,U) \quad \textnormal{almost surely}
\end{equation*}
\item For $i=0, 1$ There exist a $d_{\mathcal{X} \times \mathcal{X}}-\mu \otimes \mu$-Lebesgue differentiable function on $\mathcal{X} \times \mathcal{X}$ denoted $\Delta_{(r_i,r_i)}$ such that
   \begin{equation*}
\Esp(\sigma(Y(i),V(i)) \vert X,U)=\Delta_{(r_i,r_i)}(X,U) \quad \textnormal{almost surely}.
    \end{equation*}
\end{enumerate}
The hypothesis ($\textnormal{H}^\prime$.1) and ($\textnormal{H}^\prime$.2) imply hypothesis (H.1) together with the conclusion of Proposition \ref{prop:espconradon}. Finally, note that the existence of continuous functions representing respectively the conditional expectations $\Esp(\sigma(Y(0),V(1) \vert X, U)$ and $\Esp(\sigma(Y(i),V(i) \vert X, U)$ $i=1,0$ implies hypothesis ($\textnormal{H}^\prime$.1) and ($\textnormal{H}^\prime$.2).
\end{remark}

\begin{remark}
    An approach through regular conditional distributions also work to show that $\IPB$ is a well defined mathematical object. We only briefly sketch it. Assume that regular conditional distributions of $Y(0)$ with respect to $X$ and $V(1)$ with respect to $U$ exist (see for instance \cite{Chang1997} for sufficient conditions for the existence of such objects) and denote them respectively by $\kappa^0$ and $\kappa^1$. Since $(X,Y(0))$ and $(U,V(1))$ are independent, the function $(x,u) \mapsto  \int_{\mathcal{Y} \times \mathcal{Y}} \sigma(y^0,v^1) \kappa^0_x(dy^0) \kappa^1_u(dv^1)$ is a representative of  $\Esp(\sigma(Y(0),V(1)) \vert X,U)$. Moreover, if $\nu^0$ and $\nu^1$ are  conditional distributions of $Y(0)$ with respect to $X$ and $V(1)$ with respect to $U$ then for every mesurable subset $A$ of $\mathcal{Y}$, $\kappa_x^0(A)=\nu_x^0(A)$ and $\kappa_x^1(A)=\nu_x^1(A)$ for $\mu$-almost every $x \in \mathcal{X}$. This implies that the function $x \mapsto  \int_{\mathcal{Y} \times \mathcal{Y}} \sigma(y^0,v^1) \kappa^0_x(dy^0) \kappa^1_x(dv^1)$ is well-defined on $\mathcal{X}$ up to a set of measure zero. This prove that $\IPB$ is well defined.
\end{remark}

\subsubsection{Properties}

We now look at the specific properties of the $\Delta_{(r_i,r_j)}$. Since $(X,Y(0),Y(1))$ and $(U,V(0),V(1))$ are independent and identically distributed, it follows that for any representative $\varphi_{(r_0,r_1)}(x,u)$ and $\varphi_{(r_1,r_0)}(x,u)$ of the conditionnal expectations $\Esp(\sigma(Y(i),V(j)) \vert X,U)$ and $\Esp(\sigma(Y(j),V(i)) \vert X,U)$, we have for $i=0,\,1$ and $j=0, \, 1$
\begin{align}\label{eq:relfond}
    \varphi_{(r_i,r_j)}(x,u)=-\varphi_{(r_j,r_i)}(u,x) \quad \mu \otimes \mu-a.e.
\end{align}
It follows from \eqref{eq:relfond} that the conditional expectation $\Esp(\sigma(Y(1),V(0)) \vert X,U)$ can be represented by the Lebesgue differentiable function defined by
\begin{equation*}
   \Delta_{(r_1,r_0)}(x,u):=-\Delta_{(r_0,r_1)}(u,x). 
\end{equation*}

\begin{lemma}\label{lem:annulation_delta}
For every $x \in \mathcal{X}$, $\Delta_{(r_i,r_i)}(x,x)=0$ \, for $i=0,1$.
\end{lemma}

\begin{proof}
Apply equation \eqref{eq:relfond}  with $i=j=0$ or $i=j=1$. Since, $\Delta_{(r_0,r_0)}$ and $\Delta_{(r_1,r_1)}$ are assumed to be Lebesgue differentiable, we have that for every $(x,u) \in \supp(\mu \otimes \mu)$ 
   \begin{align*}
        \Delta_{(r_i,r_i)}(x,u)=-\Delta_{(r_i,r_i)}(u,x).
   \end{align*}
Setting $x=u$, this implies $\Delta_{(r_i,r_i)}(x,x)=0$.
\end{proof}

\begin{lemma} \label{lem:form_can} 
Let $r$ and $s$ be two ITRs. Then
  \begin{align*}
    \Esp(\sigma(Y(r),V(s)) \vert X,U)&=r(X)s(U)\Delta_{(r_1,r_1)}(X,U)+ s(U)(1-r(X))\Delta_{(r_0,r_1)}(X,U)\\
    &+s(U)(1-r(X))\Delta_{(r_0,r_1)}(X,U)+(1-s(U))r(X) \Delta_{(r_1,r_0)}(X,U)\\
    & + (1-r(X))(1-s(U)) \Delta_{(r_0,r_0)}(X,U) \quad a.s.
\end{align*}  
\end{lemma}

\begin{proof}

Writing $\mathbf{1}_{\{b \succ a\}}$ for the indicator function of the set $\{(a,b) \vert \, b \succ a\} \subset \mathcal{Y} \times \mathcal{Y}$, we obtain

\begin{align*}
   \mathbf{1}_{\{b \succ a\}}(Y(r),V(s))&=r(X)s(U) \mathbf{1}_{\{b \succ a\}}(Y(1),V(1)) 
   +(1-r(X))s(U) \mathbf{1}_{\{b \succ a\}}(Y(0),V(1))\\
   &+r(X)(1-s(U)) \mathbf{1}_{\{b \succ a\}}(Y(1),V(0))+(1-r(X))(1-s(U)) \mathbf{1}_{\{b \succ a\}}(Y(0),V(0)).
\end{align*}
Similarly,
\begin{align*}
   \mathbf{1}_{\{b \prec a\}}(Y(r),V(s))&=r(X)s(U) \mathbf{1}_{\{b \prec a\}}(Y(1),V(1)) 
   +(1-r(X))s(U) \mathbf{1}_{\{b \prec a\}}(Y(0),V(1))\\
   &+r(X)(1-s(U)) \mathbf{1}_{\{b \prec a\}}(Y(1),V(0))+(1-r(X))(1-s(U)) \mathbf{1}_{\{b \prec a\}}(Y(0),V(0)).
\end{align*}
Thus,
\begin{align*}
    \Esp(\sigma(Y(r),V(s)) \vert X,U)&=r(X)s(U)\Delta_{(r_1,r_1)}(X,U)+(1-r(X))s(U)\Delta_{(r_0,r_1)}(X,U)\\
    &+r(X)(1-s(U)) \Delta_{(r_1,r_0)}(X,U) + (1-r(X))(1-s(U)) \Delta_{(r_0,r_0)}(X,U) \; a.s.
\end{align*}
\end{proof}
Hence, in view of Lemma \ref{lem:form_can}, it is natural to set
\begin{align*}
\Delta_{(r,s)}(x,u):=& r(x)s(u)\Delta_{(r_1,r_1)}(x,u)+ s(u) (1-r(x))\Delta_{(r_0,r_1)}(x,u)\\
    &+ (1-s(u)) r(x)\Delta_{(r_1,r_0)}(x,u)+ (1-r(x))(1-s(u)) \Delta_{(r_0,r_0)}(x,u).
\end{align*}
Setting $x=u$ and using Lemma \ref{lem:annulation_delta}, we get
\begin{align*}
  \Delta_{(r,s)}(x,x)=\lbrack s(x)-r(x) \rbrack \Delta_{(r_0,r_1)}(x,x).   
\end{align*}
This leads to the following definition.
\begin{definition}\label{prop:expr_delta_rs} Let $r$ and $s$ be two ITRs. For every $x \in \supp(\mu)$,
  \begin{align*}
  \IPB_{(r,s)}(x):=\lbrack s(x)-r(x) \rbrack \Delta_{(r_0,r_1)}(x,x).   
\end{align*}  
\end{definition}
We now provide the proof of Propositions \ref{prop:opt} and \ref{prop:comp_bin}.
\begin{proof}[Proof of Proposition \ref{prop:opt}]
Let $r$ be an ITR. We want to prove that $\AIPB_{(r,s^\opt)}\geq 0$. For that purpose, it is sufficient to prove that $\IPB_{(r,s^\opt)}\geq 0$. It follows from Proposition \ref{prop:expr_delta_rs} that
\begin{align*}
    \IPB_{(r,s^\opt)}(x)&=(s^\opt(x)-r(x)) \Delta_{r_0,r_1}(x,x)\\
                        &\geq-\mathds{1}_{\{\IPB_{(r_0,r_1)}\leq 0\}} \IPB_{(r_0,r_1)}(x)\geq 0,
\end{align*}
which proves the claim.
\end{proof}

\begin{proof} [Proof of Proposition \ref{prop:comp_bin}]
Notice that
\begin{equation*}
    \mathds{1}_{ \{ V(1) > Y(0) \}}-\mathds{1}_{\{ V(1) < Y(0)\}}=V(1)-Y(0).
\end{equation*}
This implies that 
\begin{align*}
\Esp(\mathds{1}_{\{V(1) > Y(0) \}}-\mathds{1}_{\{ V(1) < Y(0)\}}|X,U)&=\Esp(V(1)-Y(0)|X,U)\\
&=\Esp(V(1)|X,U)-\Esp(Y(0)|X,U)\\
&=\Esp(V(1)|U)-\Esp(Y(0)|X).\\
\end{align*}
Then, $\Delta_{(r_0,r_1)}(x,u)=\Esp(V(1)|U=u)-\Esp(Y(0)|X=x)$. Since the pairs $(X,Y(1))$ and $(U,V(1))$ have the same distribution, one has
$\Esp(Y(1)\vert X=x)=\Esp(V(1)\vert U=x)$ at  $\mu$-almost all $x$. Hence, setting $x=u$, it follows that
\begin{equation*}
    \IPB_{(r_0,r_1)}(x)=\Esp(Y(1)|X=x)-\Esp(Y(0)|X=x) \quad \textnormal{for } \mu\textnormal{-almost all } x.
\end{equation*}
which proves the claim.
\end{proof}

\subsection{Relation with the proportion in favor of treatment} \label{sec:propinfavor}

The proportion in favor of treatment, that we denote $\Gamma$, has been introduced in \cite{Buyse:2010aa} in the context of a randomized clinical trial with stratification. It can be formalized as follows. In this case, we assume that the random variable $X$ encoding the characteristics of a patient is a discrete with values in $\{1,\ldots, K\}$. Using the notation introduced in the previous subsection, we get
\begin{equation*}
    \Gamma=\Esp(\sigma(Y(0),V(1))\vert \{X=U\}).
\end{equation*}
In particular, 
\begin{equation*}
\Gamma=\Esp(\Esp(\sigma(Y(0),V(1)) \vert X,U) |  \{X=U\} ),
\end{equation*}
that is 
\begin{equation*}
    \Gamma=\Esp(\Delta_{(r_0,r_1)}(X,U) |  \{X=U\} )
\end{equation*}
Setting 
\begin{align*}
    p_k&=\Prob(X=k),\\ 
    q_k^+&=\Prob(\sigma(Y(0),V(1))=1 |X=U=k),\\
    q^-_k&=\Prob(\sigma(Y(0),V(1))=-1 |X=U=k),
\end{align*}
we get
\begin{equation}\label{eq:gamma_disc}
    \Gamma=\frac{\sum_{k=1}^K p_k^2 [ q^+_k-q^-_k]}{\sum_{k=1}^K p_k^2}
\end{equation}
which is the estimand of the estimator proposed in \cite[Section 5.1]{Buyse:2010aa}.

The notion of proportion in favor of treatment can be generalized to the setting of pairs of ITR $(r,s)$ and continuous predictors as follows. We set $D_\varepsilon :=\{(x,u) \in  \mathcal{X} \times \mathcal{X} \vert \, d_\mathcal{X}(x,u) \leq \varepsilon \}$.

\begin{definition}
The proportion in favor of the rule $s$ with respects to the rule $r$ is, if it exists, the quantity
\begin{equation*}
    \Gamma_{(r,s)}= \lim_{\varepsilon \to 0} \Esp(\sigma(Y(r),V(s)) \vert \, (X,U) \in D_\varepsilon).
\end{equation*}
Using an approach similar to the one of Section \ref{subsec:opt}, we define an optimality criterion using $\Gamma_{(r,s)}$. That is we say that an ITR $s$ is $\Gamma$-pairwise optimal if for every ITR $r$, $\Gamma_{(r,s)} \geq 0$.
\end{definition}

We observe that
\begin{align*}
   \Gamma_{(r,s)}&= \lim_{\varepsilon \to 0} \Esp(\Esp(\sigma(Y(r),V(s)) \vert X,U) \vert \, (X,U) \in D_\varepsilon)\\
                 &=\lim_{\varepsilon \to 0} \Esp(\Delta_{(r,s)} (X,U) \vert \, (X,U) \in D_\varepsilon).
\end{align*}

\begin{remark}
As already explained our setting encompass  the case of discrete covariables. For instance, it is possible to recover the notion of proportion in favor of treatment from the notion of proportion in favor of the rule by setting $r=r_0$ and $s=r_1$ and by assuming that $X$ is a discrete random variable with value in $\mathcal{X}=\{1,\ldots, K\}$ endowed with a distance $d_\mathcal{X}$ (for instance the discrete distance).  
\end{remark}

We  now provide closed formulas for $\Gamma_{(r,s)}$ in some special case. In the discrete case, assuming that $X$ and $U$ are discrete random variables with values in $\mathcal{X}=\{1,\ldots, K\}$ endowed with the discrete distance and setting $p_k=\Prob(X=k)$, we get
\begin{align*}
    \Gamma_{(r,s)}&= \lim_{\varepsilon \to 0} \Esp(   [s(U)-r(X)]\Delta_{(r_0,r_1)}(X,U) | (X,U) \in D_\varepsilon)\\
    &= \Esp(   [s(U)-r(X)]\Delta_{(r_0,r_1)}(X,U) | X=U)\\
    &=\dfrac{\sum_{k=1}^K p_k^2[s(k)-r(k)]\Delta_{(r_0,r_1)}(k,k)}{\sum_{k=1}^K p_k^2}.
\end{align*}
Here, the second equality follows from Lemma \ref{lem:form_can} together with Lemma \ref{lem:annulation_delta}.

We now turn our attention to the continuous case. We assume that $\mathcal{X} \subset \R^d$ and that the distance $d_\mathcal{X}$ is induced by an Euclidean norm $\norm{\cdot}$ i.e., $d_\mathcal{X}(x,x^\prime)=\norm{x-x^\prime}$. We further assume that $\mathcal{X} \times \mathcal{X}$ is endowed with the norm $\norm{(x,u)}=\sqrt{\norm{x }^2+\norm{u}^2}$. These assumptions hold till the end of the section. We need a few preparatory lemmas. 
\begin{lemma}\label{lem:reduc_optim}
Assume that the  law of $X$ and $U$ is a Radon probability measure $\mu$ with support $\mathcal{X}$.  Assume in addition that $\Delta_{(r_0,r_1)}(x,u)$ and $\Delta_{(r_i,r_i)}(x,u)$ for $i=0,1$ are uniformly continuous on the support of $\mu \otimes \mu$. Then
\begin{enumerate}[(i)]
\item $\Gamma_{(r,s)}$ exists  if and only if $\lim_{\varepsilon \to 0} \Esp(   [s(U)-r(X)]\Delta_{(r_0,r_1)}( X,U ) | (X,U) \in D_\varepsilon)$ exists. If either of this two quantities exists then
\begin{equation*}
\Gamma_{(r,s)}= \lim_{\varepsilon \to 0} \Esp(   [s(U)-r(X)]\Delta_{(r_0,r_1)}( X,U ) | (X,U) \in D_\varepsilon),
\end{equation*}
\item If $s=\mathds{1}_{\{\IPB_{(r_0,r_1)}>0\}}$, then $\Gamma_{(r,s)} \geq 0$.
\end{enumerate}
\end{lemma}

\begin{proof}
\noindent (i) 
We set 
\[
\Lambda_\varepsilon=\Esp(   s(U)\,(1-r(X))\Delta_{(r_0,r_1)}(X,U) + (1-s(U)) \, r(X) \Delta_{(r_1,r_0)}(X,U) | (X,U) \in D_\varepsilon)
\]
and
\[
\Gamma_\varepsilon=\Esp((s(U)-r(X)) \Delta_{(r_0,r_1)}(X,U) \vert (X,U) \in D_\varepsilon).
\]
Using Lemma \ref{lem:form_can}, we note that 
\begin{align} \label{eq:decomplim}
\Esp(\sigma(Y(r),V(s)) \vert \, (X,U) \in D_\varepsilon) &= \Lambda_\varepsilon + \Esp(r(X)s(U)\Delta_{(r_1,r_1)}(X,U) \vert  (X,U) \in D_\varepsilon) \\
&+\Esp((1-r(X))(1-s(U)) \Delta_{(r_0,r_0)}(X,U) \vert  (X,U) \in D_\varepsilon).
\end{align}
Hence, to prove $(i)$, it is sufficient to show that
\begin{align}
    &\lim_{\varepsilon \to 0} \Esp(r(X)s(U)\Delta_{(r_1,r_1)}(X,U) \vert  (X,U) \in D_\varepsilon)=0, \label{eq:lim1}\\
    &\lim_{\varepsilon \to 0} \Esp((1-r(X))(1-s(U)) \Delta_{(r_0,r_0)}(X,U) \vert  (X,U) \in D_\varepsilon)=0, \label{eq:lim2}\\
    & \lim_{\varepsilon \to 0} (\Lambda_\varepsilon-\Gamma_\varepsilon) = 0. \label{eq:lim3}
\end{align}

Indeed, Equations \eqref{eq:decomplim}, \eqref{eq:lim1} and \eqref{eq:lim2} together imply that $\lim_{\varepsilon \to 0} \Esp(\sigma(Y(r),V(s)) \mid (X,U) \in D_\varepsilon)$ exists if and only if $\lim_{\varepsilon \to 0} \Lambda_\varepsilon$ exists, and that, if they exist, they are equal. Similarly, Equation \eqref{eq:lim3} implies that $\lim_{\varepsilon \to 0} \Lambda_\varepsilon$ exists if and only if $\lim_{\varepsilon \to 0} \Gamma_\varepsilon$ exists, and, if so, they are equal. Consequently, this establishes that if any of the three limits exists, then they all exist and are equal.

We prove Equality \eqref{eq:lim1}. Equality \eqref{eq:lim2} will not be proved, as its proof follows a similar approach.
Notice that 
\begin{align*}
\norm{(x,u)-(u,u)} &= \norm{(x-u,0)}\\
                      &= \norm{x-u}.
\end{align*}
 Let $\eta>0$. Since $\Delta_{(r_1,r_1)}$ is uniformly continuous on the support of $\mu \otimes \mu$, there exists $\delta>0$ such that for every $(x,u)$ and $(x^\prime,u^\prime) \in \mathcal{X} \times \mathcal{X}$ such that $\norm{(x,u)-(x^\prime,u^\prime)} \leq \delta$, 
\begin{equation*}
|\Delta_{(r_1,r_1)}(x,u)-\Delta_{(r_1,r_1)}(x^\prime,u^\prime)| \leq \eta.
\end{equation*}
It follows from Lemma \ref{lem:annulation_delta} that $\Delta_{(r_1,r_1)}(u,u)=0$. Hence, if $(x,u) \in D_\delta$, one has
\begin{equation*}
    |\Delta_{(r_1,r_1)}(x,u)|=|\Delta_{(r_1,r_1)}(x,u)-\Delta_{(r_1,r_1)}(u,u)| \leq \eta.
\end{equation*}
This imples that for $\varepsilon \leq \delta$, we have
\begin{align*}
    |\Esp(r(x)s(u)\Delta_{(r_1,r_1)}(x,u) \vert  (X,U) \in D_\varepsilon)&  \leq \frac{1}{\mu \otimes \mu(D_\varepsilon)} \int_{D_\varepsilon}|\Delta_{(r_1,r_1)}(x,u)| d\mu(x) \otimes d\mu(u)\\
    &\leq \eta.
\end{align*}

We now study the limit as $\varepsilon$ tends to zero of  $\Lambda_\varepsilon-\Gamma_\varepsilon$. 
Observe that
\begin{align*}
   \Lambda_\varepsilon-\Gamma_\varepsilon&= \Esp ( s(U)\,(1-r(X))\Delta_{(r_0,r_1)}(X,U) + (1-s(U)) \, r(X) \Delta_{(r_1,r_0)}(X,U)\\
   & \quad - [s(U)-r(X)]\Delta_{(r_0,r_1)}( X,U )  | (X,U) \in D_\varepsilon )\\
   &= \Esp ( r(X) (\Delta_{(r_0,r_1)}(X,U) -\Delta_{(r_0,r_1)}(U,X)) \\
   &\quad-s(U) r(X) \lbrack \Delta_{(r_0,r_1)}(X,U)-\Delta_{(r_0,r_1)}(U,X) \rbrack | (X,U) \in D_\varepsilon )\\
   &=\Esp ( r(X) (1-s(U)) (\Delta_{(r_0,r_1)}(X,U) -\Delta_{(r_0,r_1)}(U,X))  | (X,U) \in D_\varepsilon ).
\end{align*}
Thus, we have
\begin{align*}
  &|\Esp ( r(X) (1-s(U)) (\Delta_{(r_0,r_1)}(X,U) -\Delta_{(r_0,r_1)}(U,X))  | (X,U) \in D_\varepsilon )|\\
  &\leq \Esp( | \Delta_{(r_0,r_1)}(X,U)-\Delta_{(r_0,r_1)}(U,X) | \vert (X,U) \in D_\varepsilon).
\end{align*}
By hypothesis $\Delta_{(r_0,r_1)}$ is uniformly continuous on $\mathcal{X} \times \mathcal{X}$ and $\norm{(x,u)-(u,x)}=\norm{(x-u,u-x)}=2 \norm{x-u}$. Hence, for every $\eta>0$, we have for $\varepsilon$ sufficiently small, $|\Delta_{(r_0,r_1)}(x,u)-\Delta_{(r_0,r_1)}(u,x)| \leq \eta$. This implies that $\Esp( | \Delta_{(r_0,r_1)}(X,U)-\Delta_{(r_0,r_1)}(U,X) | \vert (X,U) \in D_\varepsilon) \leq \eta$ which proves that $\lim_{\varepsilon \to 0} (\Lambda_\varepsilon-\Gamma_\varepsilon) = 0$. This implies that $\lim_{\varepsilon \to 0} \Lambda_\varepsilon$ exists if and only if $\lim_{\varepsilon \to 0}\Gamma_\varepsilon$ exists and, if so, they are equal.
This conclude the proof.\\

\noindent (ii) We have
\begin{align*}
    \Gamma_\varepsilon &=\Esp((s(U)-r(X)) \Delta_{(r_0,r_1)}(X,U) \vert (X,U) \in D_\varepsilon)\\
     &=\frac{1}{\mu \otimes \mu(D_\varepsilon)} \int_{D_\varepsilon} [s(u)-r(x)] \Delta_{(r_0,r_1)}(x,u)  \,d\mu(x) \otimes d\mu(u)\\
    &= \frac{1}{\mu \otimes \mu(D_\varepsilon)} \int_{D_\varepsilon} [s(u)-r(x)]\left( \Delta_{(r_0,r_1)}(x,u) -\Delta_{(r_0,r_1)}(u,u) \right)  \,d\mu(x) \otimes d\mu(u)\\
    & \quad + \frac{1}{\mu \otimes \mu(D_\varepsilon)} \int_{D_\varepsilon} [s(u)-r(x)] \Delta_{(r_0,r_1)}(u,u)  \,d\mu(x) \otimes d\mu(u).
\end{align*}
Let $\delta>0$. Since $\Delta_{(r_0,r_1)}$ is uniformly continuous, there exist $\varepsilon$ such that $\norm{x-u} < \varepsilon$ implies 
\begin{equation*}
\left \vert \Delta_{(r_0,r_1)}(x,u) -\Delta_{(r_0,r_1)}(u,u) \right \vert \leq \delta.
\end{equation*}
Hence,
\begin{align*}
    &\left \vert  \int_{D_\varepsilon} [s(u)-r(x)]\left( \Delta_{(r_0,r_1)}(x,u) -\Delta_{(r_0,r_1)}(u,u) \right)  \,d\mu(x) \otimes d\mu(u) \right \vert \\
    \leq & \int_{D_\varepsilon} \left\vert [s(u)-r(x)]\left( \Delta_{(r_0,r_1)}(x,u) -\Delta_{(r_0,r_1)}(u,u) \right) \right\vert  \,d\mu(x) \otimes d\mu(u)\\
    \leq& \mu \otimes \mu(D_\varepsilon) \delta.
\end{align*}
This implies
\begin{equation*}
    \lim_{\varepsilon \to 0} \frac{1}{\mu \otimes \mu(D_\varepsilon)}  \int_{D_\varepsilon} [s(u)-r(x)]\left( \Delta_{(r_0,r_1)}(x,u) -\Delta_{(r_0,r_1)}(u,u) \right)  \,d\mu(x) \otimes d\mu(u) =0.
\end{equation*}
We now study the sign of
\begin{equation*}
 \Gamma^+_\varepsilon  =\frac{1}{\mu \otimes \mu(D_\varepsilon)} \int_{D_\varepsilon} [s(u)-r(x)] \Delta_{(r_0,r_1)}(u,u)  \,d\mu(x) \otimes d\mu(u).
\end{equation*}

Notice that
\begin{align*}
 [s(u)-r(x)] \Delta_{(r_0,r_1)}(u,u)=
\begin{cases}
   (1-r(x)) \, \Delta_{(r_0,r_1)}(u,u) \quad \text{if} \; \Delta_{(r_0,r_1)}(u,u) > 0,\\
   -r(x) \, \Delta_{(r_0,r_1)}(u,u) \quad \text{if} \; \Delta_{(r_0,r_1)}(u,u)\leq 0
\end{cases}
\end{align*}
which implies that for every $(x,u) \in \mathcal{X} \times \mathcal{X}$, $\Gamma^+_\varepsilon \geq 0$ and $\lim_{\varepsilon \to 0} \Gamma_\varepsilon = \lim_{\varepsilon \to 0} \Gamma_\varepsilon^+$ which, implies the result.
\end{proof}

We now provide sufficient conditions for the existence of $\Gamma_{(r,s)}$ when the probability measure $\mu$ admits a density with respect to the Lebesgue measure.

\begin{lemma}\label{lem:lebesgue_diff_product}
 Consider the normed space $(\R^d,\norm{\cdot})$ endowed with a Radon measure $\mu$.
 Let $x\in \R^n$ and $f,g \colon \R^n \to \R$ be two measurable functions such that $x$ is a Lebesgue point of these two functions. Assume that $g$ is bounded in a neigbourhood of $x$. Then $x$ is a Lebesgue point of $fg$  
\end{lemma}

\begin{proof}
    We notice that in a neigbourhood of $x$,
    \begin{align*}
        |f(u)g(u)-f(x)g(x)|&=|f(u)g(u)-f(x)g(u)+f(x)g(u)-f(x)g(x)|\\
                           &\leq |g(u)(f(u)-f(x)|+ |f(x)(g(u)-g(x)|\\
                           &\leq C |f(u)-f(x)|+ |f(x)| |g(u)-g(x)|.
    \end{align*}
Since $x$ is a Lebesgue point of both $f$ and $g$ it follows that  $\lim_{\varepsilon \to 0} \frac{1}{\mu(B(x,\varepsilon))} \int_{B(x,\varepsilon)} |g(u)-g(x)| du = 0$ and $\lim_{\varepsilon \to 0} \frac{1}{\mu(B(x,\varepsilon))} \int_{B(x,\varepsilon)} |f(u)-f(x)| du=0$ which implies that $\lim_{\varepsilon \to 0} \frac{1}{\mu(B(x,\varepsilon))} \int_{B(x,\varepsilon)} |f(u)g(u)-f(x)g(x)| du = 0$
\end{proof}

\begin{theorem}\label{thm:sqr_density}
Let the law of $X$ and $U$ be a Radon probability measure $\mu$ on the normed space $(\R^d,\norm{\cdot})$. Suppose that $\mu$ is compactly supported and has a bounded density $f$ with respects to the Lebesgue measure. Furthermore, assume that $\Delta_{(r_0,r_0)}$, $\Delta_{(r_1,r_1)}$ and $\Delta_{(r_0,r_1)}$ are continuous on the support of the product measure $\mu \otimes \mu$. Let $r$ and $s$ be two measurable individualized treatment rules. Then $\Gamma_{(r,s)}$ is well defined and

\begin{equation}
     \Gamma_{(r,s)}=\dfrac{1}{\int_{\R^n} f^2(u) du} \int_{\R^n} [s(u)-r(u)]\Delta_{(r_0,r_1)}(u,u) f^2(u) du.
\end{equation}

\end{theorem}

\begin{proof}
Since the support of $\mu$ is compact, it follows that the restriction of $\Delta_{(r_0,r_0)}$ and $\Delta_{(r_1,r_1)}$ to $\supp(\mu \otimes \mu)$ are uniformly continuous. Hence, applying Lemma \ref{lem:reduc_optim}, we are reduce to compute the limit $\lim_{\varepsilon \to 0^+} \Gamma_\varepsilon$ where
\begin{equation*}
    \Gamma_\varepsilon=\Esp(   [s(U)-r(X)]\Delta_{(r_0,r_1)}(X,U) | (X,U) \in D_\varepsilon).
\end{equation*}

Let $\varepsilon >0$. Then

\begin{equation*}
     \Gamma_\varepsilon=\dfrac{ \int_{D_\varepsilon} [s(u)-r(x)] \Delta_{(r_0,r_1)}(x,u) f(x) f(u) du\, dx}{\int_{D_\varepsilon} f(u)f(x) du \, dx}
\end{equation*}
We consider the following change of variables $(x=\alpha + \beta, u=\alpha-\beta)$. Then 
\begin{align*}
D_\varepsilon&=\{(x,u) ; \norm{x-u} \leq \varepsilon \}=\{(\alpha,\beta) ; \norm{\beta} \leq \frac{\varepsilon}{2}\}\simeq \R^n \times B \left(0,\frac{\varepsilon}{2}\right).
\end{align*}

This also implies that
\begin{align*}
\Gamma_\varepsilon & =\frac{\int_{\R^n}  \left( \int_{B(0,\frac{\varepsilon}{2})} [s(\alpha-\beta)-r(\alpha+\beta)] \Delta_{(r_0,r_1)}(\alpha+\beta,\alpha-\beta) f(\alpha+\beta) f(\alpha-\beta) d \beta \right) \, d\alpha     }{\int_{\R^n}  \left( \int_{B(0,\frac{\varepsilon}{2})} f(\alpha+\beta)  f(\alpha-\beta) d \beta \right) \, d\alpha     }\\
                 &= \frac{\int_{\R^n}  \left(  \left( \int_{B(0,\frac{\varepsilon}{2})}  d \beta \right)^{-1} \int_{B(0,\frac{\varepsilon}{2})} [s(\alpha-\beta)-r(\alpha+\beta)] \Delta_{(r_0,r_1)}(\alpha+\beta,\alpha-\beta) f(\alpha+\beta) f(\alpha-\beta) d \beta \right) \, d\alpha     }{\int_{\R^n}  \left( \left( \int_{B(0,\frac{\varepsilon}{2})}  d \beta \right)^{-1} \int_{B(0,\frac{\varepsilon}{2})} f(\alpha+\beta)  f(\alpha-\beta) d \beta \right) \, d\alpha     }.
\end{align*}
Let $\alpha \in \R^n$ and consider the function
\begin{equation*}
    g_\alpha(\beta)= [s(\alpha-\beta)-r(\alpha+\beta)] \Delta_{(r_0,r_1)}(\alpha+\beta,\alpha-\beta) f(\alpha+\beta) f(\alpha-\beta).
\end{equation*}
The functions $\beta \mapsto r(\alpha+\beta)$, $\beta \mapsto s(\alpha+\beta)$, $\beta \mapsto f(\alpha+\beta)$, $\beta \mapsto f(\alpha-\beta)$ are Lebesgue differentiable at $\beta=0$ for almost every $\alpha$ thanks to the Lebesgue differentiation theorem. Moreover, $\beta \mapsto \Delta_{(r_0,r_1)}(\alpha+\beta,\alpha-\beta)$ is Lebesgue differentiable at $\beta=0$ as it is continuous. This, together with Lemma \ref{lem:lebesgue_diff_product}, implies that for almost every $\alpha \in \R^n$, $g_\alpha$ is Lebesgue differentiable in $\beta=0$. 

Consider the function
\begin{equation*}
    h(\alpha,\varepsilon)=\frac{1}{\int_{B(0,\frac{\varepsilon}{2})}  d \beta}\int_{B(0,\frac{\varepsilon}{2})} g_\alpha(\beta)  d \beta.
\end{equation*}
We now proceed in two step. First, using the Lebesgue differentiability of $g_\alpha$, we compute $\lim_{\varepsilon \to 0} h(\varepsilon,\alpha)$. Second, using the dominated convergence theorem, we prove that we can interchange the limit along $\varepsilon$ and the integration against $\alpha$.

By the Lebesgue differentiability of $g_\alpha(\beta)$ at zero, for almost every $\alpha \in \R^n$
\begin{align*}
  \lim_{\varepsilon \to 0} h(\alpha,\varepsilon) &= g_\alpha(0)\\
                                                 &=[s(\alpha)-r(\alpha)]\Delta_{(r_0,r_1)}(\alpha,\alpha) f^2(\alpha).
\end{align*}
We can further assume that $\varepsilon \leq 2$, thus $\norm{\beta}\leq 1$. As the support of $f$ is compact, we consider the set 
\begin{equation*}
K=\{x \in \R^n ; d(x, \supp(f)) \leq 1\}
\end{equation*}
where $d(x,\supp(f))=\inf_{z \in \supp(f)} \norm{x-z}$. Hence, if $\alpha \notin K$, $h(\alpha,\varepsilon)=0$. This implies
\begin{align*}
 |h(\alpha,\varepsilon)| & \leq \frac{1}{\int_{B(0,\frac{\varepsilon}{2})}  d \beta}\int_{B(0,\frac{\varepsilon}{2})} |g_\alpha(\beta)|  d \beta \; \mathds{1}_K(\alpha)\\
                        & \leq \frac{1}{\int_{B(0,\frac{\varepsilon}{2})}  d \beta}\int_{B(0,\frac{\varepsilon}{2})} \norm{f}_\infty d \beta \; \mathds{1}_K(\alpha)\\
                        &\leq \norm{f}_\infty \; \mathds{1}_K(\alpha).
\end{align*}
    Hence, by the Lebesgue's dominated convergence theorem
\begin{align*}
\lim_{\varepsilon \to 0} \int_{\R^n} h(\alpha,\varepsilon) d \alpha= \int_{\R^n} [s(\alpha)-r(\alpha)]\Delta_{(r_0,r_1)}(\alpha,\alpha) f^2(\alpha) \,d \alpha
\end{align*}
A similar argument shows that
\begin{equation*}
\lim_{\varepsilon \to 0} \frac{1}{ \int_{B(0,\frac{\varepsilon}{2})}  d \beta}\int_{\R^n}  \int_{B(0,\frac{\varepsilon}{2})} f(\alpha+\beta)  f(\alpha-\beta) d \beta  \, d\alpha =    \int_{\R^n} f^2(\alpha) \, d\alpha.
\end{equation*}

The theorem follows immediately.
\end{proof}

\begin{remark}
\begin{enumerate}[(i)]
\item Using the closed formula obtained for $\Gamma_{(r,s)}$ in the discrete or continuous setting, we remark that  rules which are pairwise optimal are $\Gamma$-pairwise optimal.
\item Though the quantity $\Gamma_{(r,s)}$ allows to define a notion of optimality for ITRs, it existence requires stronger assumption than the one required for the existence of the $\AIPB$ (Equation \eqref{eq:AIPB}) and the pairwise optimality criterion.
\item The main difference between the two setting come from the underlying sampling mechanism. In the case of $\AIPB$, the sampling mechanism roughly amount to sample an individual of type $x$ and duplicate him/her, whereas in the $\Gamma$-pairwise setting, the sampling mechanism roughly amount to sample pair of the form $(x,x)$ among all possible pairs of this form.
\end{enumerate}
\end{remark}

\section{Pointwise consistency of two sample conditional \texorpdfstring{$U$}{U}-statistics}\label{sec:conv_knn}

In this section, we provide a proof of Theorem \ref{thm:cv_knn_Delta}. In fact, we prove a more general version which applies to several two-sample $U$-statistics estimator relying on $k$-nearest neighbors.

 We work on the normed  vector space $(\R^d,\norm{\cdot})$ and assume that space of pairs $(x,u)$ is the normed space $(\R^d \times \R^d,\norm{\cdot}_\pi)$ where $\norm{(x,u)}_{\pi}=\max(\norm{x},\norm{u})$. We introduce the following notations.  We denote by $\mu$  the distribution of $X$ and $U$. Since $X$ and $U$ are iid the distribution $\mu_{X,U}$ of $(X,U)$ is given by the product measure $\mu \otimes \mu$. For $x \in \R^d$, we let $((X_{(1)}(x),Y_{(1)}(x)),\ldots,(X_{(m)}(x),Y_{(m)}(x)))$ be a reordering of the data $((X_1,Y_1),\ldots ,(X_m,Y_m))$ according to increasing values of $\norm{X_i-x}$, that is
 \begin{equation*}
     \norm{X_{(1)}(x)-x} \leq \ldots \leq \norm{X_{(m)}(x)-x} 
 \end{equation*}
 where potential ties i.e., $\norm{X_i-x}=\norm{X_j-x}$ with $i \neq j$ are broken according to the following rule. If $\norm{X_i-x}=\norm{X_j-x}$  with $i \neq j$ then $X_i$ will be deemed  closer to $x$ that $X_j$ if and only if $i <j$. Notice that $Y_{(i)}(x)$ is the $Y_i$ of the couple $(X_i,Y_i)$. We apply the same notational conventions for the elements $(U_j,V_j)$ of the experimental group $E$. 
 
  We let $\sigma \colon \R^\ell \times \R^\ell \to \R$ be a Borel mesurable function and seek to estimates $T(x,u)=\Esp(\sigma(Y,V) \mid X=x, U=u)$.

 We consider nearest neighbor type estimators. Let $(v_{mi,nj})_{i,j}$ with $1 \leq i \leq m$, $1 \leq j \leq n$, such that for every $i, j$, $v_{mi,nj} \geq 0$ and $\sum_{i,j} v_{mi,nj}=1$. We consider estimators of the form
 \begin{equation*}
 T_N(x,u): =\sum_{i,j} v_{mi,nj} \sigma(Y_{(i)}(x),V_{(j)}(u))
 \end{equation*}
where $N=m+n$.
 
We have the following result.
 \begin{theorem}\label{thm:cv_knn}
    Let  $\sigma \colon \R^l \times \R^l \to \R$ be a Borel function such that $\sigma(Y,V)$ is bounded and such that $T(x,u)=\Esp(\sigma(Y,V) \mid X=x,U=u)$ is continuous on the support of $\mu \otimes \mu$. 
    
    Let $N=n+m$ and assume that there exist sequences of integers ${c}=\{c_m\}$ and $e=\{e_n\}$ such that

 \begin{enumerate}[(i)]
        \item $\dfrac{n}{m} \underset{N \to \infty}{\longrightarrow} \lambda \neq 0$, \label{cond:relnm},
        \item $c_m \underset{N \to \infty}{\longrightarrow} \infty$ and $e_n \underset{N \to \infty}{\longrightarrow} \infty$.
        \item $\dfrac{c_m}{N} \underset{N \to \infty}{\longrightarrow} 0$ and $\dfrac{e_n}{N} \underset{N \to \infty}{\longrightarrow} 0$ \label{cond:distknn},
        \item $v_{mi,nj}=0$ when $i>c_m$ or $j>e_n$, 
        \item $\sup_{m,n}(c_m \, e_n  \max_{i,j} (v_{mi,nj})) < \infty$ \label{cond:conv}.
    \end{enumerate}
Then the corresponding two-samples conditionnal $U$-statistics $T_N$ satisfies
\begin{equation*}
\Esp|T_N(x,u)-T(x,u) |^2 \to 0 \; \textnormal{at all} \;(x,u)\; in \supp(\mu \otimes \mu). 
\end{equation*}
In particular, for all $(x,u) \in \supp(\mu \otimes \mu)$,
\begin{equation*}
    T_N(x,u) \to T(x,u) \; \textnormal{in probability}.
\end{equation*}
\end{theorem}

\begin{remark}
 We wish to emphasize that the above theorem does not only hold for almost every point but for all points in the support of the distribution. This necessitates stronger assumptions than those typically employed. Such a result is required as we restrict $\hat{\Delta}$ to the set $\{(x , x) \in \R^d \times \R^d \}$ to estimate the $\IPB$.  It is worth noting that this set may potentially have measure zero. Consequently, a convergence result for the estimator $\widehat{\Delta}$ that only holds for $\mu \otimes \mu$-almost every $(x,u)$ would not be informative regarding the estimation of the $\IPB$.
\end{remark}
To prove Theorem \ref{thm:cv_knn}, we will need the following lemmas.
\begin{lemma}
Let $(m_\nu)_{\nu \in \N}$ and $(n_\nu)_{\nu \in \N}$ be two sequences of elements of $\N$ such that $m_\nu \underset{\nu \to \infty}{\longrightarrow}  \infty$ and $n_\nu \underset{\nu \to \infty}{\longrightarrow} \infty$  and let $N_\nu=m_\nu+n_\nu$. Assume that $\dfrac{n_\nu}{m_\nu} \underset{\nu \to \infty}{\longrightarrow} \lambda \neq 0$. Then
\begin{enumerate}[(i)]\label{lem:trivial_fact}
    \item There exists a constant $\beta$ such that for every $k \in \N$, $N\geq \beta k$ implies that $\min(n,m) \geq k$.
    \item If $N \underset{\nu \to \infty}{\longrightarrow} \infty$, then $m \underset{\nu \to \infty}{\longrightarrow} \infty$ and $n \underset{\nu \to \infty}{\longrightarrow} \infty$.
\end{enumerate}
\end{lemma}

\begin{proof}
    It is clear that $(i)$ implies $(ii)$. We now prove $(i)$. Since $\dfrac{n_\nu}{m_\nu} \underset{\nu   \to \infty}{\longrightarrow} \lambda$, we assume that $\nu$ is sufficiently large so that $\mid\frac{n_\nu}{m_\nu}-\lambda\mid \leq \frac{\lambda}{2}$. This implies that
    \begin{equation*}
        \frac{\lambda}{2} m_\nu \leq n_\nu \leq \frac{3 \lambda}{2}m_\nu
    \end{equation*}
Thus $ \frac{2}{3\lambda+2} N_\nu\leq  m_\nu$ and $ \frac{\lambda}{\lambda+2} N_\nu \leq n_\nu$. Set $\beta=\max(\frac{3\lambda+2}{2},\frac{\lambda+2}{\lambda})$ which is well defined as $\lambda \neq 0$. Let $k \in \N$. If $N_\nu \geq \beta k$, then $\min(m_\nu,n_\nu) \geq k$.

\end{proof}

\begin{lemma}\label{lem:rhs_main}
Let $p \geq 1$ and let $g \colon \R^d \times \R^d \to \R$ be a bounded Borel function, continuous  on the support of $\mu \otimes \mu$ and such that $\Esp|g(X,U)|^p) < \infty$. Let $N=n+m$ and assume that there exists sequences of integers ${c}=\{c_m\}$ and $e=\{e_n\}$ such that

 \begin{enumerate}[(i)]
        \item $\dfrac{n}{m} \underset{N \to \infty}{\longrightarrow} \lambda \neq 0$, \label{cond:relnmrhs_main}
        \item $c_m \underset{N \to \infty}{\longrightarrow} \infty$ and $e_n \underset{N \to \infty}{\longrightarrow} \infty$,
        \item $\dfrac{c_m}{N} \underset{N \to \infty}{\longrightarrow} 0$ and $\dfrac{e_n}{N} \underset{N \to \infty}{\longrightarrow} 0$, \label{cond:knn}
        \item $v_{mi,nj}=0$ when $i>c_m$ or $j>e_n$, \label{cond:supp}
        \item $\sup_{m,n}(c_m \, e_n  \max_{i,j} v_{mi,nj}) < \infty$.
    \end{enumerate}
Then
\begin{equation*}
\Esp \left\vert \sum_{ \substack{1 \leq i \leq m \\ 1 \leq j \leq n}} v_{mi,nj} g(X_{(i)}(x),U_{(j)}(u)) -g(x,u) \right\vert^p \to 0 \; \textnormal{for} \; (x,u) \in \supp(\mu \otimes \mu). 
\end{equation*}
\end{lemma}

\begin{proof}
    We adapt the proof of \cite[Lemma 11.1]{Biau_Devroye2015} to our setting. We endow $\R^d \times \R^d$ with the norm
    $\norm{(x,u)}_\pi=\max(\norm{x},\norm{u})$.
    
    Since $p \geq 1$, it follows from Jensen's inequality and conditions \eqref{cond:knn}-\eqref{cond:supp} that for some constant $\alpha$
\begin{align*}
    \Esp \left\vert \sum_{ \substack{1 \leq i \leq m \\ 1 \leq j \leq n}} v_{mi,nj} g(X_{(i)}(x),U_{(j)}(u))-g(x,u) \right\vert^p &\leq \Esp \left[ \sum_{ \substack{1 \leq i \leq m \\ 1 \leq j \leq n}}  v_{mi,nj} \vert g(X_{(i)}(x),U_{(j)}(u))-g(x,u) \vert ^p \right]\\
    & \leq  \frac{\alpha}{ce} \Esp \left[ \sum_{ \substack{1 \leq i \leq m \\ 1 \leq j \leq n}}  \vert g(X_{i},U_{j})-g(x,u) \vert ^p \mathds{1}_{\{ \Gamma_i \leq c \}} \mathds{1}_{ \{\Sigma_j \leq e \}}  \right]
\end{align*}
where $\Gamma_i$ (resp. $\Sigma_j$) is the rank of $X_i$ (resp. $U_j$) with respect to the distances to $x$ (resp. $u$). Potential ties are broken by index comparisons according to the rule defined at the beginning of Section \ref{sec:conv_knn}.

Let $(x,u)$ be a point in the support $S=\supp(\mu \otimes \mu)$ of $\mu \otimes \mu$. Let $\varepsilon >0$. Since $g$ is continuous on $S$ there exists $\delta>0$ such that for every $(a,b) \in S \cap B((x;u),\delta)$, $\vert g(a,b)-g(x,u) \vert \leq \varepsilon$.

For every $k \in \N$, we set
\begin{equation*}
    \mathcal{C}^\delta_k=\{\omega \in \Omega \mid \forall (m , n) \in \N \times \N \, \textnormal{such that} \, \min(m,n) \geq k\,; \max(\norm{X_{(c)}(x)-x} \,,\norm{U_{(e)}(u)-u})<\delta\}.
\end{equation*}

The points $(X_i,U_j)$ belong almost surely to the support of $\mu \otimes \mu$. Condition \eqref{cond:relnmrhs_main}-\eqref{cond:knn} ensure that $\frac{c_m}{m}\underset{m \to \infty}{\longrightarrow} 0$ and $\frac{e_n}{n}\underset{n \to \infty}{\longrightarrow} 0$. It follows from \cite[Lemma 2.2]{Biau_Devroye2015} that $\norm{X_{(c)}(x)-x} \underset{m \to \infty}{\longrightarrow} 0$ and $\norm{U_{(e)}(u)-u} \underset{n \to \infty}{\longrightarrow} 0$ almost surely. This implies that $\Prob((\mathcal{C}^{\delta}_k)^c)\underset{k \to \infty }{\longrightarrow} 0$.

Let $k_0 \in \N$ such that for every $k \geq k_0$, $\Prob((\mathcal{C}^{\delta}_k)^c) \leq \varepsilon$. Assume that $N \geq \beta k_0$ with $\beta$ as in Lemma \ref{lem:trivial_fact}. One has
\begin{align*}
    \vert g(X_{i},U_{j})-g(x,u) \vert^p \mathds{1}_{\{ \Gamma_i \leq c \}} \mathds{1}_{ \{\Sigma_j \leq e \}} \leq \left(\varepsilon + \norm{g}_\infty \mathds{1}_{(\mathcal{C}^{\delta}_k)^c} \right) \mathds{1}_{\{ \Gamma_i \leq c \}} \mathds{1}_{ \{\Sigma_j \leq e \}}. 
\end{align*}
This implies
\begin{align*}
\frac{\alpha}{ce} \Esp \left[ \sum_{ \substack{1 \leq i \leq m \\ 1 \leq j \leq n}}  \vert g(X_{i},U_{j})-g(x,u) \vert ^p \mathds{1}_{\{ \Gamma_i \leq c \}} \mathds{1}_{ \{\Sigma_j \leq e \}}  \right] & \leq \alpha \varepsilon+  \frac{\alpha \norm{g}_\infty}{ce} \, \Esp \left[ \sum_{ \substack{1 \leq i \leq m \\ 1 \leq j \leq n}} \mathds{1}_{(\mathcal{C}^{\delta}_k)^c} \mathds{1}_{\{ \Gamma_i \leq c \}} \mathds{1}_{ \{\Sigma_j \leq e \}} \right]\\
& \leq \alpha \varepsilon+  \frac{\alpha \norm{g}_\infty}{ce} \, \sum_{ \substack{1 \leq i \leq m \\ 1 \leq j \leq n}} \Esp  \mathds{1}_{(\mathcal{C}^{\delta}_k)^c} \Esp ( \mathds{1}_{\{ \Gamma_i \leq c \}} \mathds{1}_{ \{\Sigma_j \leq e \}} )\\
&\leq \alpha \varepsilon+  \alpha \, \norm{g}_\infty \, \Prob((\mathcal{C}^{\delta}_k)^c)\\
&\leq \alpha \,(1+\norm{g}_\infty) \, \varepsilon.
\end{align*}
This conclude the proof. 
\end{proof}

\begin{lemma}\label{lem:vanishing_condexp}
    Let $\sigma \colon \R^l \times \R^l \to \R$ be a borel mesurable function. Let $T(X,U)=\Esp(\sigma(Y,V) \mid X, U)$. 
For $1 \leq i \leq m$, $1 \leq j \leq n$, set 
\begin{equation*}
Z_{(i)(j)}(x,u)=\sigma(Y_{(i)}(x),V_{(j)}(u))-T(X_{(i)}(x),U_{(j)}(u)).
\end{equation*}
For $1 \leq  i, p \leq m$ and $1 \leq j,q \leq n$ with $i \neq p$ and $j \neq q$, the random variables $Z_{(i)(j)}(x,u)$ and $Z_{(p)(q)}(x,u)$ are independent conditional on $X_1,\ldots X_m,U_1,\ldots,U_n$. Moreover,
\begin{equation*}
    \Esp(Z_{(i)(j)}(x,u))\mid X_1,\ldots X_m,U_1,\ldots,U_n)=0.
\end{equation*}
\end{lemma}

\begin{proof}
The proof is similar to the one of \cite[Proposition 8.1]{Biau_Devroye2015}. %
\end{proof}

\begin{proof}[Proof of Theorem \ref{thm:cv_knn}]
Our proof is an adaptation of the proof of \cite[Theroem 11.1]{Biau_Devroye2015}.
It follows from Jensen inequality that    
\begin{equation}\label{eq:thm_conv}
    \begin{split}
    \Esp \left| \sum_{i,j} v_{mi,nj} \sigma(Y_{(i)}(x),V_{(j)}(u) )- T(x,u) \right|^2 & \leq 2 \, \Esp \left| \sum_{i,j} v_{mi,nj} \lbrack\sigma(Y_{(i)}(x),V_{(j)}(u) )- T(X_{(i)}(x),U_{(j)}(u)) \rbrack \right|^2\\
    &\;+2 \,  \Esp \left| \sum_{i,j} v_{mi,nj} T(X_{(i)}(x),U_{(j)}(u) )- T(x,u) \right|^2.
    \end{split}
\end{equation}
The second term of the right-hand side of the above inequality is converging to zero by Lemma \ref{lem:rhs_main}. For the first term, we have
\begin{align*}
    \Esp \left| \sum_{i,j} v_{mi,nj} \lbrack\sigma(Y_{(i)}(x),V_{(j)}(u) )- T(X_{(i)}(x),U_{(j)}(u)) \rbrack \right|^2& \leq    \Esp \left( \sum_{i,j} v_{mi,nj} Z_{(i)(j)}(x,u) \right) ^2\\
    \hspace{-5cm} &\hspace{-3cm}= \Esp\left( \sum_{i,j,p,q} v_{mi,nj} v_{mp,nq} \, Z_{(i)(j)}(x,u) \,  Z_{(p)(q)}(x,u)  \right).
\end{align*}
 Moreover,

 \begin{equation}\label{eq:sum_decomposition}
    \begin{split}
    \sum_{i,j,p,q} v_{mi,nj} v_{mp,nq} \, Z_{(i)(j)}(x,u) \,  Z_{(p)(q)}(x,u) 
     &= \sum_{\substack{i,j}} v_{mi,nj}^2  \, Z^2_{(i)(j)}(x,u)\\
    &+ \sum_{\substack{i\\j \neq q}} v_{mi,nj} v_{mi,nq}   \,  Z_{(i)(j)}(x,u) \, Z_{(i)(q)}(x,u)  \\
    &+  \sum_{\substack{i\neq p \\j}} v_{mi,nj} v_{mp,nj}   \, Z_{(i)(j)}(x,u) \, Z_{(p)(j)}(x,u)  \\
    &+  \sum_{\substack{i\neq p \\j \neq q}} v_{mi,nj} v_{mp,nq}   \, Z_{(i)(j)}(x,u) \, Z_{(p)(q)}(x,u).
    \end{split}
\end{equation}

It follows from Lemma \ref{lem:vanishing_condexp}, that for $i\neq p$, $j \neq q$ $Z_{(i)(j)}(x,u)$ and   $Z_{(p)(q)}(x,u)$ are independent conditionally to $X_1, \ldots, X_m,U_1,\ldots,U_n$. Hence, for $i\neq p$, $j \neq q$
\begin{align}\label{eq:vanishing_expectation}
\Esp \left( Z_{(i)(j)}(x,u) \, Z_{(p)(q)} (x,u) \right) &=0.
\end{align}
Moreover, since $\sigma(Y,V)$ is bounded, there exists $M>0$ such that for every pair $(i,j)$, $\vert Z_{(i)(j)}(x,u) \vert \leq M $. Combining this together with the equation \eqref{eq:sum_decomposition} and \eqref{eq:vanishing_expectation}
Hence, we get that
\begin{align}\label{eq:ineq convergence}
     \Esp \left( \sum_{i,j} v_{mi,nj} Z_{(i)(j)}(x,u) \right) ^2 &\leq  M \left[ \sum_{\substack{i,j}} v_{mi,nj}^2  
     +  \sum_{\substack{i\\j \neq q}} v_{mi,nj} v_{mi,nq}   + \sum_{\substack{i\neq p \\j}} v_{mi,nj} v_{mp,nj} \right]. 
\end{align}
We now control each term of the sum of the right-hand side. We have
\begin{equation*}
\sum_{\substack{i,j}} v_{mi,nj}^2 \leq \max_{i,j}(v_{ij}),
\end{equation*}
and
\begin{align*}
    \sum_{\substack{i\\j \neq q}} v_{mi,nj} v_{mi,nq}&= \sum_{i=1}^m  \sum_{j \neq q} \mathds{1}_{\{j \leq e_n\}} v_{mi,nj} v_{mi,nq}\\
                                                     &\leq \max_{i,j}(v_{mi,nj}) \sum_{i=1}^m  \sum_{j \neq q} \mathds{1}_{\{j \leq e_n\}} v_{mi,nq}\\
                                                     &\leq \max_{i,j}(v_{mi,nj}) \sum_{i=1}^m  \sum_{j , q} \mathds{1}_{\{j \leq e_n\}} v_{mi,nq}\\
                                                     &\leq \max_{i,j}(v_{mi,nj}) e_n.
 \end{align*}
Similarly, we prove that
\begin{align*}
    \sum_{\substack{i\neq p \\j}} v_{mi,nj} v_{mp,nj}   &\leq \max_{i,j}(v_{mi,nj}) c_m.
\end{align*}
Thanks to assumption \eqref{cond:conv}, the above inequalities shows that each of the terms in the right-hand side of inequality \eqref{eq:ineq convergence} converges to zero, when $N \to \infty$. This proves that the left-hand side of inequality \eqref{eq:thm_conv} converges to zero when $ N \to \infty$ which proves the theorem.
\end{proof}

\section{Bagging}\label{sec:bagging}

In this section, we provide a proof of Theorem \ref{thm:bagging}.

\begin{proof}[Proof of Theorem \ref{thm:bagging}]
Observe that the random variable $Z$ is independent of $\mathcal{D}_{m,n}$.
To prove the theorem, it is sufficient to prove that if $g_K(x,u, \mathcal{D}^\prime_K)$ is pointwise consistent in $(x,u)$ then $g_{L}(\cdot,\mathcal{D}^\prime_{m,n}(Z))$ is pointwise consistent at $(x,u)$. Note that we have the following identity:
\begin{equation*}
    g_L(x,u,\mathcal{D}^\prime_{m,n}(Z))=\sum_{\ell=1}^{k} \sum_{\substack{\alpha \in \mathcal{I}(\ell,m)\\ \beta \in \mathcal{I}_i(\ell,n)}} \, \mathds{1}_{\{(\alpha,\beta)\}}(Z) \, g_\ell(x,u,\mathcal{D}^\prime_{\alpha,\beta}).
\end{equation*}
Let $\delta>0$. The above identity implies
\begin{align*}
    \Prob(\vert g_L(x,u,\mathcal{D}^\prime_{m,n}(Z))-\Delta(x,u) \vert > \delta) &= \sum_{\ell=1}^{k} \sum_{\substack{\alpha \in \mathcal{I}(\ell,m)\\ \beta \in \mathcal{I}_i(\ell,n)}} \, \Prob(\mathds{1}_{\{(\alpha,\beta)\}}(Z) \,  \mid g_\ell(x,u,\mathcal{D}^\prime_{\alpha,\beta})-\Delta(x,u) \mid > \delta )\\
    & =\sum_{\ell=1}^{k} \sum_{\substack{\alpha \in \mathcal{I}(\ell,m)\\ \beta \in \mathcal{I}_i(\ell,n)}} \,
    \Prob(Z=(\alpha,\beta))  \Prob(\mid g_\ell(x,u,\mathcal{D}^\prime_{\alpha,\beta})-\Delta(x,u) \mid > \delta ).
\end{align*}
Let $ \varepsilon >0$. There exists $\ell_0 \in \N$ such that for $\ell \geq \ell_0$, $\Prob( \vert g_\ell(x,u,\mathcal{D}^\prime_{\alpha,\beta}) - \Delta(x,u) \vert > \delta) \leq \varepsilon$. Assuming $k\geq \ell_0$, we have

\begin{equation}\label{ineq:globalbag}
\Prob(\vert g_L(x,u,\mathcal{D}^\prime_{m,n}(Z))-\Delta(x,u) \vert > \delta)  \leq   \varepsilon \sum_{\ell=\ell_0}^{k} \sum_{\substack{\alpha \in \mathcal{I}(\ell,m)\\ \beta \in \mathcal{I}_i(\ell,n)}} 
    \! \! \Prob(Z=(\alpha,\beta)) + \sum_{\ell=1}^{\ell_0} \sum_{\substack{\alpha \in \mathcal{I}(\ell,m)\\ \beta \in \mathcal{I}_i(\ell,n)}} \! \! \Prob(Z=(\alpha,\beta)).\\
\end{equation}
Moreover,
\begin{equation}\label{ineq:firstbag}
    \sum_{\ell=\ell_0}^{k} \sum_{\substack{\alpha \in \mathcal{I}(\ell,m)\\ \beta \in \mathcal{I}_i(\ell,n)}} \Prob(Z=(\alpha,\beta))  \leq 1
\end{equation}
and
\begin{align}
    \sum_{\ell=1}^{\ell_0} \sum_{\substack{\alpha \in \mathcal{I}(\ell,m)\\ \beta \in \mathcal{I}_i(\ell,n)}} \Prob(Z=(\alpha,\beta)) &=\sum_{\ell=1}^{\ell_0} \sum_{\substack{\alpha \in \mathcal{I}(\ell,m)\\ \beta \in \mathcal{I}_i(\ell,n)}} \Prob(Z=(\alpha,\beta) \mid L=\ell) \Prob(L=\ell) \notag\\
    &= \sum_{\ell=1}^{\ell_0} \binom{k}{\ell} q_k^\ell (1-q_k)^{k-l} \sum_{\substack{\alpha \in \mathcal{I}(\ell,m)\\ \beta \in \mathcal{I}_i(\ell,n)}} \Prob(Z=(\alpha,\beta) \mid L=\ell) \notag\\
    &= \sum_{\ell=1}^{\ell_0} \binom{k}{\ell} q_k^\ell (1-q_k)^{k-l} \underset{k \to \infty}{\longrightarrow} 0. \label{ineq:secondbag}
\end{align}
Combining the expressions \eqref{ineq:globalbag}, \eqref{ineq:firstbag} and \eqref{ineq:secondbag}, we get that for $k$ sufficiently large
\begin{equation*}
    \Prob(\vert g_L(x,u,\mathcal{D}^\prime_{m,n}(Z))-\Delta(x,u) \vert > \delta)  \leq  2 \varepsilon,
\end{equation*}
which proves the claim.
\end{proof}

\section{Simulation} \label{sec:datgen}

In this section, we provide a detailed description of the data-generating mechanism used in our simulations.
We first describe how we generate the population then how we generate the outcome. 

Our approach follows the one of \cite{Grolleau2024}. The population we generate is described via eight correlated covariates, one normal, three log-normal, and four Bernoulli. We first define a Gaussian random vector and transform it to obtain the desired random variables. We proceed as follows.

\begin{enumerate}
    \item  We define a diagonal matrix $D=\mathrm{diag}(\lambda_1\ldots,\lambda_8)$, where $\lambda_i=1+0.3 \times i$.
    \item We draw an orthogonal $O$ matrix from the Haar distribution on the orthogonal group $\mathcal{O}(8)$ and define the covariance matrix:
    \begin{equation*}
        \Sigma=O \, D \, O^t.
    \end{equation*}
    \item We generate a sample of size $2n$, where either $n=200$ or $n=1000$, by drawing from a Gaussian vector $(Z_1,\ldots,Z_8) \sim \mathcal{N}(0,\Sigma)$ and generates $(X_1,\ldots,X_8)$ as follows:
    \begin{align*}
        X_1&=Z_1 \\ 
        (X_2,X_3,X_4)&=(\exp(Z_2),\exp(Z_3),\exp(Z_4)) \\ 
        (X_4,X_5,X_6,X_8)&=(\mathds{1}_{\{Z_5>0\}},\ldots,\mathds{1}_{\{Z_8>0\}}).
    \end{align*}
    Finally, we form the vector $X=(X_0,X_1,\ldots,X_8)$ with $X_0\equiv1$ to  allow for an intercept.
\end{enumerate}

\subsection{Scenario 1: the case of two binary outcomes}
In this scenario, we assume that the potential outcome $Y(0)=(Y^1(0),Y^2(0))$ of a patient $X$ and $V(1)=(V^1(1),V^2(1))$  of a patient $U$  are random binary vectors generated as follows. The distribution of the potential outcomes are 
\begin{align*}
    Y(0)|X &\sim \Multinomial(N=1,p(X)=(p_{11}(X),p_{10}(X),p_{01}(X),p_{00}(X)),\\
    V(1)|X &\sim \Multinomial(N=1,q(U)=(q_{11}(U),q_{10}(U),q_{01}(U),q_{00}(U))
\end{align*}

with
\begin{equation*}
    p_{ab}(U)= \dfrac{\exp(\alpha_{ab}^t X)}{\sum \exp(\alpha_{ab}^t X) } \quad \text{and} \quad q_{ab}(U)= \dfrac{\exp(\beta_{ab}^t U)}{\sum \exp(\beta_{ab}^t U) }
\end{equation*}
where $a,b \in \{0,1\}$, $\alpha_{00}=\beta_{00}=0$ and the parameter $(\alpha_{11}^t,\alpha_{10}^t,\alpha_{01}^t,\beta_{11}^t,\beta_{10}^t,\beta_{01}^t)$ is drawn from the distribution $\mathcal{U}\left([-1,1]^{6\times 9}\right)$.

\subsection{Scenario 2: the case of a binary and a numerical outcome}
In this scenario, we assume that the potential outcome $Y(0)=(Y^1(0),Y^2(0))$ and $V(1)=(V^1(1),V^2(1))$ a patients $X$ and $U$ are random vectors where $Y^1(0)$ and $V^1(0)$ are a Bernoulli random variables and $Y^2(0)$ and $V^2(1)$ are  numerical outcomes with integer values between 0 and 25. The potential outcomes are generated as follows.
\begin{align*}
    Y^1(0) | X \sim \Ber(p_1(X)) \quad \text{and} \quad  Y^2(0) | X,Y^1(0) \sim \Binom(25, p_2(X,Y^1(0)), \\
    V^1(1) | X \sim \Ber(q_1(U)) \quad \text{and} \quad  V^2(1) | U,V^1(1) \sim \Binom(25, q_2(U,V^1(1))
\end{align*}
with 
\begin{align*}
p_1(X)=\dfrac{1}{1+\exp(\alpha^t X)} \quad \text{and} \quad p_2(X,Y^1(0))=\dfrac{1}{1+\exp(\beta^t X + Y^1(0) \times \gamma^t X)},\\
q_1(U)=\dfrac{1}{1+\exp(\theta^t X)} \quad \text{and} \quad q_2(U,V^1(1))=\dfrac{1}{1+\exp(\xi^t X + V^1(1) \times \rho^t X)}
\end{align*}
where $(\alpha^t,\beta^t,\gamma^t,\theta^t,\xi^t,\rho^t)$ is drawn from the distribution $\mathcal{U}\left([-1,1]^{6\times 9}\right)$.

\subsection{Oracle IPB}

We evaluate each classifier on a dataset of 2 000 0000 rows. For each rows, we have computed the oracle value of $\IPB_{(r_0,r_1)}$ thanks to the following formulas.

\begin{align*}
\Delta_{(r_0,r_1)}(x,u) = \int_{\mathcal{Y} \times \mathcal{Y}} \sigma(y,v) K_x(dy) \, K_u(dv)    
\end{align*}
where $K_x(dy)$ and $K_u(dv)$ designates respectively the conditional laws of $Y(0) | X$ and $V(1) | U$.
In the case of the first scenario, this leads to
\begin{equation*}
    \Delta_{(r_0,r_1)}(x,u)= \sum_{y,v \in \{0,1\}^2} \sigma(y,v) p_y(x)q_v(u),
\end{equation*}
while in the second scenario, it yields
\begin{equation*}
    \Delta_{(r_0,r_1)}(x,u)= \sum_{y,v \in S} \binom{25}{y^2} \binom{25}{v^2}\sigma(y,v)  p_1(x)p_2(x,y^1) q_1(u)q_2(u,v^1),
\end{equation*}
where $y=(y^1,y^2)$, $v=(v^1,v^2)$ and $S=\{0,1\} \times \{0, \ldots, 25\}$.

\section{List of personalization variables} \label{sec:covar}
The model used to construct the ITR from the IST-3 uses as outomes the variables \texttt{ohs6} and \texttt{euroqol6} data. It takes into account the following 47 categorical random variable and the following 14 continuous covariates for personalization. We refer the reader to the IST-3 data dictionnary for the description of these variables.

\begin{enumerate}
        \item \texttt{gender}, 
        \item \texttt{livealone\_rand}: Lived alone before stroke ?,
        \item \texttt{indepinadl\_rand}: Independent in activities of daily living before stroke ?,
        \item \texttt{infarct} : Recent ischaemic change likely cause of this stroke ?,
        \item \texttt{antiplat\_rand} : Received antiplatelet drugs in last 48 hours ?,
        \item \texttt{atrialfib\_rand} : Patient in atrial fibrillation at randomisation,
        \item \texttt{country},
        \item \texttt{liftarms\_rand}: Able to lift both arms off bed at randomisation,
        \item \texttt{ablewalk\_rand}: Able to walk without help at randomisation,
        \item \texttt{weakface\_rand}:Unilateral weakness affecting face at randomisation,
        \item \texttt{weakarm\_rand}: Unilateral weakness affecting arm or hand at randomisation,
        \item \texttt{weakleg\_rand}: Unilateral weakness affecting leg or foot at randomisation,
        \item \texttt{dysphasia\_rand}: Dysphasia at randomisation,
        \item \texttt{hemianopia\_rand}: Homonymous hemianopia at randomisation ,
        \item \texttt{visuospat\_rand}: Visuospatial disorder at randomisation,
        \item \texttt{brainstemsigns\_rand}: Brainstem or cerebellar signs at randomisation , 
        \item \texttt{otherdeficit\_rand}: Other neurological deficit at randomisation,
        \item \texttt{stroketype}: Stroke subtype,
        \item \texttt{R\_scannorm}: Completely normal  pre-randomisation scan (R scan),
        \item \texttt{R\_infarct\_territory} Infarct territory (from R scan) ,
         \item \texttt{R\_infarct\_size} : Infarct size (from R scan),
        \item \texttt{R\_hypodensity}: Degree of acute hypodensity (R scan),
        \item \texttt{R\_hyperdensity},
        \item \texttt{R\_hyperdense\_arteries}: Hyperdense arteries visible on R scan,
        \item \texttt{R\_isch\_change} : Ischaemic change visible on R scan,
         \item \texttt{R\_swelling} : Degree of tissue swelling in acute infarct (R scan) ,
        \item \texttt{R\_atrophy} : Atrophy visible on R scan,
        \item \texttt{R\_whitematter} : White matter visible on R scan,
        \item \texttt{R\_oldlesion} : Old vascular lesion visible on R scan,
        \item \texttt{R\_nonstroke\_lesion} : Non-stroke lesion visible on R scan ,
        \item \texttt{apca} : infarcts in ACA / PCA territory,
        \item \texttt{subinf1}: Acute small subcortical infarcts ?,
        \item \texttt{cbzinf1}: Acute cortical borderzone infarcts ?,
        \item \texttt{cinf1}:  Acute cerebellar infarct ?,
        \item \texttt{stem1}: Acute brainstem infarct ?,
        \item \texttt{aca1}: Anterior cerebral artery affected ?,
        \item \texttt{pca1}: Posterior cerebral artery affected ?,
        \item \texttt{aspirin\_pre}: Aspirin before admission ?,
        \item \texttt{dipyridamole\_pre}: Dipyridamole before admission ?,
        \item \texttt{clopidogrel\_pre}: Clopidogrel before admission ?,
        \item \texttt{lowdose\_heparin\_pre}: Low dose heparin before admission ?,
        \item \texttt{warfarin\_pre}: Warfarin before admission ?,
        \item \texttt{antithromb\_pre}: Other thrombotic agents before admission ?,
        \item \texttt{hypertension\_pre}: Treatment for hypertension before admission ?,
        \item \texttt{diabetes\_pre}: Treatment for diabetes before admission ?,
        \item \texttt{stroke\_pre} History of previous stroke or transient ischemic attack ?,
        \item \texttt{tmt}=1-\texttt{itt\_treat} : Allocated treatment.
\end{enumerate}

It uses the following continuous covariates:
\begin{enumerate}
\item \texttt{ageimp} : age,
\item \texttt{sbprand} : Systolic BP at randomisation (mm Hg), 
\item \texttt{dpbrand} : Diastolic BP at randomisation (mm Hg), 
\item \texttt{weight} : Estimated weight (kg),
\item \texttt{glucose} : Blood glucose (mmol/L), 
\item \texttt{gcs\_eye\_rand} : Best eye response (Glasgow Coma Scale) at randomisation, 
\item \texttt{gcs\_motor\_rand} : Best motor response (Glasgow Coma Scale) at randomisation , 
\item \texttt{gcs\_verbal\_rand} : Best verbal response (Glasgow Coma Scale) at randomisation, 
\item \texttt{gcs\_score\_rand} : Total Glasgow Coma Scale score at randomisation, 
\item \texttt{nihss} : Total NIH Stroke Score at randomisation,
\item \texttt{treatdelay} : Time (hour) from stroke to treatment,
\item \texttt{konprob}: Probability of good outcome based on Konig model, 
\item \texttt{R\_mca\_aspects} : Aspects score (max 10) for Middle cerebral artery (R scan), 
\item \texttt{R\_tot\_aspects} : Total aspects score (max 12) including ACA/PCA (R scan).
\end{enumerate}

\section{Single outcome}\label{sec:ohs6}

Here, using our methodology, we construct an ITR using the OHS at six months as the sole outcome. The score used is defined in Table \ref{tab:monolevelscore}. We denote by $r^{\opt}_m$ the optimal rule associated with this score. We report the following result.

\begin{table}[htb]
	\centering
	\caption{Scoring $\sigma_{ij}$ of a pair $(i,j)$ of individuals taken from the control ($C$) and experimental ($E$) groups, respectively, with outcomes $(Y^{1}_i)$ and  $(V^{1}_j)$.}
	\label{tab:monolevelscore}
	\begin{tabular}{lccl} 
		\hline
        \multicolumn{4}{l}{OHS at 6 months}  \\ \cline{1-3}
		$(Y^{1}_i,V^{1}_j)$ & Ordering & Label&  $\sigma_{ij}$ \\ \hline
		$V_j^1-Y_i^1>0$ &  $V^{1}_j \succ Y^{1}_i$ & Favorable &  $+ 1$ \\
		$V_j^1-Y_i^1<0$ & $V^{1}_j \prec Y^{1}_i$ & Unfavorable & $-1$\\
		$V_j^1-Y_i^1=0$ & $V^{1}_j \bowtie Y^{1}_i$ & Tie/neutral & $+ 0$\\
		\hline
	\end{tabular}
\end{table}

\begin{table}[htb]
	\centering
	\caption{Estimate of the AIPB for the rule treating everyone with rt-PA and for the estimated optimal rules together with the standard errors. The last two standard errors are conditional to the training set.}\label{result:ist3}
	\begin{tabular}{ccc}
    \hline
    \rule{0pt}{2.8ex} 
        Treatment rules & Mean & Standard error\\ \hline
		$\widehat{\AIPB}_{(r_0,r_1)}$ & 0.038 & (0.019)\\
        $\widehat{\AIPB}_{(r_1,\hat{r}_m^{\opt})}$ & 0.041 & (0.001) \\
        $\widehat{\AIPB}_{(r_0,\hat{r}_m^{\opt})}$ & 0.083 & (0.002)\vspace{0.1cm}\\
    \hline
	\end{tabular}
\end{table}

{\footnotesize
\begin{longtable}{lll}
\caption{Characteristics of patients according to the treatment recommended by $\hat{r}_m^{\opt}$. NIHSS=National Institutes of Health Stroke Scale. TACI=total anterior circulation infarct. PACI=partial anterior circulation infarct. LACI=lacunar infarct. POCI=posterior circulation infarct.}
\label{tab:char_recom} \\
\hline
& \textbf{rt-PA} & \textbf{No rt-PA} \\
& \textbf{recommended} & \textbf{recommended} \\
& \textbf{(n=1639)} & \textbf{(n=1396)} \\
\hline
\endfirsthead

\hline
& \textbf{rt-PA} & \textbf{No rt-PA} \\
& \textbf{recommended} & \textbf{recommended} \\
& \textbf{(n=1639)} & \textbf{(n=1396)} \\
\hline
\endhead
\multicolumn{3}{l}{\textbf{Baseline variables}} \\
\hline
\multicolumn{3}{l}{\textbf{Age (years)}} \\
18--50 & 4.4\% &  3.9\% \\
51--60 & 6.5\% & 6.8\%  \\
61--70 & 12.3\% & 11.7\%\\
71--80 & 22.6\% & 25.4\%  \\
81--90 & 47.8\% & 44.7\% \\
$>$90 & 6.5\% & 7.4\% \\
\hline
\multicolumn{3}{l}{\textbf{Sex}} \\
Female & 53.1\% & 50.1\% \\
\hline
\multicolumn{3}{l}{\textbf{NIHSS}} \\
0--5 & 20.6\% & 19.6\% \\
6--10 & 28.8\% & 27.2\% \\
11--15 & 19.4\% & 20.3\% \\
16--20 & 17.8\% & 18.0\% \\
$>$20 &  13.4\% & 14.8\% \\
\hline
\multicolumn{3}{l}{\textbf{Delay in randomisation}} \\
0--3.0 h & 21.4\% & 19.3\% \\
3.0--4.5 h & 37.3\% & 38.4\% \\
4.5--6.0 h & 37.3\% & 38.5\% \\
$>$6.0 h & 4.0\% & 3.8\% \\
\hline
Atrial fibrillation & 30.1\% & 30.2\% \\
\hline
\multicolumn{3}{l}{\textbf{Systolic blood pressure}} \\
$\leq$143 mm Hg & 35.1\% & 31.5\% \\
144--164 mm Hg & 31.3\% & 33.4\% \\
$\geq$165 mm Hg & 33.6\% & 35.1\% \\
\hline
\multicolumn{3}{l}{\textbf{Diastolic blood pressure}} \\
$\leq$74 mm Hg & 34.1\% & 32.8\% \\
75--89 mm Hg & 33.5\% & 34.3\% \\
$\geq$90 mm Hg & 32.4\% & 32.9\% \\
\hline
\multicolumn{3}{l}{\textbf{Blood glucose}} \\
$\leq$5 mmol/L & 21.1\% & 18.2\% \\
6--7 mmol/L & 46.3\% & 49.3\% \\
$\geq$8 mmol/L & 32.6\% & 32.5\% \\
\hline
Treatment with antiplatelet drugs in previous 48 h & 51.3\% & 51.6\% \\
\hline
\multicolumn{3}{l}{\textbf{Predicted probability of poor outcome at 6 months}} \\
$<$40\% & 54.2\% & 54.6\% \\
40--50\% & 11.1\% & 10.3\% \\
50--75\% & 23.1\% & 26.0\% \\
$\geq$75\% & 11.5\% & 9.1\% \\
\hline
\multicolumn{3}{l}{\textbf{Stroke clinical syndrome}} \\
TACI & 42.2\% & 44.1\% \\
PACI & 38.7\% & 36.6\% \\
LACI & 11.0\% & 10.8\% \\
POCI & 7.9\% & 8.4\% \\
Other & 0.2\% &0.1\% \\
\hline
\multicolumn{3}{l}{\textbf{Baseline variables collected from prerandomisation scan}} \\
\multicolumn{3}{l}{\textbf{Expert reader's assessment of acute ischaemic change}} \\
Scan completely normal & 9.2\% & 9.2\% \\
Scan not normal but no sign of acute ischaemic change & 50.9\% & 50.1\% \\
Signs of acute ischaemic change & 39.9\% & 40.7\% \\
\hline
\end{longtable}
}
\end{document}